\documentclass[12pt,a4paper]{article}
\usepackage[utf8]{inputenc}
\usepackage[T1]{fontenc}
\usepackage[english]{babel}
\usepackage{amsmath,amssymb,amsthm}
\usepackage{mathtools}
\usepackage{hyperref}
\usepackage{geometry}
\usepackage{enumitem}
\usepackage{booktabs}
\usepackage{xcolor}
\theoremstyle{plain}
\newtheorem{theorem}{Theorem}[section]
\newtheorem*{theorem*}{Theorem}
\newtheorem{proposition}[theorem]{Proposition}
\newtheorem{lemma}[theorem]{Lemma}

\theoremstyle{definition}
\newtheorem{definition}[theorem]{Definition}
\newtheorem{example}[theorem]{Example}
\newtheorem{conjecture}[theorem]{Conjecture}
\theoremstyle{remark}
\newtheorem{remark}[theorem]{Remark}
\newcommand{\classP}{\mathbf{P}}
\newcommand{\classNP}{\mathbf{NP}}
\newcommand{\PeqNP}{\classP = \classNP}
\newcommand{\PneqNP}{\classP \neq \classNP}
\newcommand{\SAT}{\mathsf{SAT}}
\newcommand{\KC}{K}
\newcommand{\Kbt}{K^t}
\newcommand{\Kt}{\mathrm{Kt}}
\newcommand{\gam}{\gamma}
\newcommand{\gP}{\gamma_{\classP}}
\newcommand{\gC}[1]{\gamma_{#1}}
\newcommand{\gCFG}{\gamma_{\mathrm{CFG}}}
\newcommand{\TU}{T_U}
\newcommand{\TA}[1]{T_{#1}}
\newcommand{\Tad}{T_{\mathrm{ad}}}
\newcommand{\OCout}{\mathrm{OC}_{\mathrm{out}}}
\newcommand{\Comp}{\mathcal{C}}
\newcommand{\gU}[1]{\gamma_{#1}}
\newcommand{\poly}{\mathrm{poly}}
\newcommand{\Tdec}{T_{\mathrm{dec}}}
\newcommand{\Tder}{T_{\mathrm{derive}}}
\newcommand{\Hs}{H_{\mathrm{s}}}
\title{\textbf{Witness Complexity of Short Descriptions:\\
A Cryptographic Perspective}}
\author{
  Fabio F.G.\ Buono\\
  \small Independent Researcher\\
  \small \href{https://orcid.org/0009-0004-9199-2793}{ORCID: 0009-0004-9199-2793}
}
\date{\today}

\begin{document}
\maketitle
\begin{abstract}
  In cryptographic practice, a short key or certificate is useful only if it
  can be decompressed or verified within an acceptable time budget; a compact
  representation that requires superpolynomial work to expand offers no
  operational guarantee within a bounded-time protocol. This paper formalises
  that gap by introducing \emph{witness complexity} $\gam(x)$, the minimum
  running time over all near-shortest descriptions of a string $x$ on a
  universal Turing machine.
  
  The quantity $\gam(x)$ is distinct from both Shannon entropy and Kolmogorov
  complexity $\KC(x)$: an object may have low descriptive complexity yet high
  $\gam(x)$, if its near-shortest descriptions are computationally expensive
  to execute. We establish five main results: invariance of $\gam$ up to
  polynomial factors across universal machines~(A); a conditional separation
  showing that low $\KC$ does not imply polynomial decompressibility,
  assuming $\PneqNP$~(B); an unconditional lower bound on $\gam$, assuming
  only the classical incomputability of $\KC$~(B'); an exact biconditional
  characterisation of $\PeqNP$ in terms of the class-relative variant
  $\gP$, restricted to certificates whose Kolmogorov complexity is
  commensurate with the instance size~(C); and unconditional polynomial-time
  tractability on structurally guided families of $\classNP$ instances~(D).
  To the authors' knowledge, $\gP$ is the only measure admitting such a
  biconditional characterisation of $\PeqNP$ in the standard Turing model.
  
  Part~II introduces three companion quantities measuring work per bit of
  genuine input information, overhead beyond writing the output, and
  information density of solutions, and develops their interaction with
  $\gam$. An application to grammar-based compression exhibits an
  unconditional gap between grammar size and derivation cost invisible to
  all existing measures. A falsifiable conjecture connects the framework to
  the observed tractability of industrial $\classNP$ instances. Collectively,
  the results position $\gam$ as a metric for the \emph{usability} of keys
  and certificates: low $\KC(x)$ alone is insufficient; low $\gP(x)$ is
  required for a short description to be operationally accessible within a
  bounded-time protocol.
\end{abstract}

\tableofcontents
\bigskip
\section{Introduction and Standing Hypotheses}
\label{sec:intro}

\subsection{Cryptographic motivation and the gap this paper fills}
\label{sec:intro:gap}

A central assumption in cryptographic protocol design is that a compact
representation of a key, certificate, or witness can be \emph{used} within a
bounded time budget: key schedules must be derivable in real time, certificates
must be verifiable before a session times out, and proof-carrying data must be
decompressible at the point of verification. This assumption is operationally
necessary but has not been formalised as a complexity-theoretic invariant.

Three classical measures of information have coexisted without addressing it.
Kolmogorov complexity $\KC(x)$ \cite{kolmogorov1965,liVitanyi2008} measures
the length of the shortest description of $x$, but says nothing about the
computational cost of \emph{using} that description. Shannon entropy
\cite{shannon1948} measures the average description length for a source, with
encoding and decoding costs assumed away. Chaitin's $\Omega$
\cite{chaitin1975} encodes all halting information in the extreme where no
description is computationally accessible. None of these measures asks: given
that a short description exists, how much work is required to execute it?

This paper introduces $\gam(x)$, the \emph{witness complexity} of $x$, defined
as the minimum running time over all near-shortest descriptions of $x$ on a
fixed universal prefix Turing machine (Definition~\ref{def:gamma}). The
quantity $\gam(x)$ is the answer to the question above. Its class-relative
version $\gP(x) = \gC{\classP}(x)$ restricts the decompressor to
polynomial-time machines and is the principal object of study
(Definition~\ref{def:gamma-C}).

\paragraph{Cryptographic objective.}
The cryptographic goal of this paper is to provide a formal basis for
reasoning about the \emph{usability} of compact representations. Concretely,
$\gP(x)$ being polynomial in the relevant instance size is a sufficient
condition for $x$ to be operationally accessible within a polynomial-time
protocol via a near-shortest encoding; it is also necessary among
decompressors that operate on near-shortest descriptions
(Definition~\ref{def:gamma-C}). The framework applies to three cryptographic scenarios.
\begin{enumerate}[label=(\roman*)]
  \item \emph{Key management.} A key $k$ of length $\ell$ may be stored as a
        near-shortest description $d$ with $|d| \approx \KC(k)$. The key is
        usable in a protocol only if $\gP(k)$ is polynomial in the security
        parameter; otherwise the derivation cost exceeds the protocol budget.
  \item \emph{Certificate verification.} An $\classNP$ certificate $c$ for an
        instance $w$ may be encoded compactly as a near-shortest description.
        A sufficient condition for the certificate to be recoverable within a
        polynomial-time protocol via such an encoding is that $\gP(c)$ is
        polynomial in $|w|$.
  \item \emph{Key-schedule and PRNG expansion.} A pseudorandom generator with
        seed $s$ of length $k \ll n$ produces output of length $n$. The seed
        is a near-shortest description of the output; the expansion cost
        satisfies $\gam(\text{output}) \geq n$ (at least $n$ steps are needed
        to write all output bits) and $\gam(\text{output}) \leq
        T_{\mathrm{PRNG}}(s)$, where $T_{\mathrm{PRNG}}(s)$ is the running
        time of the generator on seed $s$. The $\gam$ framework makes this
        cost explicit and comparable across constructions.
\end{enumerate}
The threat model and attack scenarios arising from large $\gP$ values are
formalised in Section~\ref{sec:threat}.

\subsection{Technical contributions}
\label{sec:intro:contributions}

We prove the following results, stated informally here and formally in
Sections~\ref{sec:defs}--\ref{sec:main}.

\begin{enumerate}[label=\textbf{(\Alph*)}]
  \item \textbf{Invariance} (Theorem~\ref{thm:A}). For any two universal
        prefix Turing machines $U_1, U_2$, there exists a polynomial $q$ such
        that $\gU{U_1}(x) \leq q(|x|) \cdot \gU{U_2}(x) + O(1)$ for all
        $x$. The polynomial factor is made explicit via a compiler lemma
        (Lemma~\ref{lem:compiler}). This invariance establishes $\gam$ as a
        machine-independent complexity measure, a prerequisite for its use in
        any protocol-independent security argument.

  \item \textbf{Conditional separation} (Theorem~\ref{thm:B}). Assuming
        $\PneqNP$, there exists an explicit infinite family $\{x_\varphi\}$
        indexed by Boolean formulas such that $\KC(x_\varphi) = O(|\varphi|)$
        and $\gP(x_\varphi)$ is superpolynomial in $|\varphi|$. Under
        $\PneqNP$, the existence of a short description does not imply
        polynomial decompressibility.

  \item \textbf{Exact characterisation of $\PeqNP$} (Theorem~\ref{thm:C}).
        In the standard multi-tape deterministic Turing model: $\PeqNP$ if
        and only if for every $L \in \classNP$ there exists a polynomial $p_L$
        such that for every instance $w$ of $L$ there exists a valid certificate
        $x_w$ with $\KC(x_w) = \Omega(|w|)$ and $\gP(x_w) \leq p_L(|w|)$.
        The $\KC(x_w) = \Omega(|w|)$ condition is necessary: for KC-poor
        certificates, $\gP = \infty$ unconditionally (Remark~\ref{rem:KC-condition}).
        The non-trivial direction ($\Leftarrow$) uses a dovetailing schedule
        (Lemma~\ref{lem:dovetail}) together with a fixed-machine argument
        that reduces the search to polynomially many candidates.

  \setcounter{enumi}{3}
  \item[\textbf{(B')}] \textbf{Unconditional lower bound} (Theorem~\ref{thm:Bprime},
        Appendix~\ref{sec:appendix-B}). Assuming only the classical
        incomputability of $\KC$: for every polynomial $p$, there exists $y$
        with $\gam(y) > p(|y|)$. This result is unconditional and
        complementary to~(B).
  \setcounter{enumi}{3}

  \item \textbf{Tractability on structured families} (Theorem~\ref{thm:D}).
        If $F$ is a structurally guided family for $L \in \classNP$
        (Definition~\ref{def:guided}): a family where a polynomial-time
        procedure $\mathcal{P}_F$ produces, for each $x \in F$, a
        near-shortest description $d$ of a valid witness $y_x$ that is
        also expandable to $y_x$ in time $\leq p(|x|)$ for a fixed
        polynomial $p$, then $L$ is solvable in polynomial time on $F$.
        The result is unconditional and does not require $\PeqNP$.
\end{enumerate}

Part~II (Sections~\ref{sec:adaptive}--\ref{sec:summary}) introduces three
companion quantities: adaptive complexity $\Tad(A,x) = \TA{A}(x)/\KC(x)$
(work per bit of genuine input information), output overhead complexity
$\OCout(A,x) = \TA{A}(x) - |A(x)|$ (overhead beyond writing the output
$A(x)$), and structural entropy $\Hs(y) = \KC(y)/\log_2 |y|$ (information
density of solutions). An application to grammar-based compression (Section~\ref{sec:grammar})
exhibits an unconditional gap between grammar size and derivation depth,
invisible to all existing measures (Lemma~\ref{lem:grammar}).

\subsection{Standing hypotheses and model of computation}
\label{sec:intro:hypotheses}

Throughout this paper the following hypotheses hold without further notice.

\begin{enumerate}[label=(H\arabic*)]
  \item \label{hyp:prefix}
        A universal prefix Turing machine $U$ is fixed once and for all.
        All Kolmogorov complexities $\KC(x)$ are defined with respect to $U$.
        We denote by $c_0$ a fixed constant depending only on $U$ (not on $x$)
        such that for every $x \in \{0,1\}^*$ there exists a self-delimiting
        program $d$ with $U(d) = x$ and $|d| \leq \KC(x) + c_0$. Such a
        constant exists because $\KC(x)$ is the infimum of program lengths and
        the infimum is achieved up to a fixed slack by the definition of the
        universal machine~\cite{liVitanyi2008}. The value of $c_0$ is chosen
        large enough to also absorb the constant overhead of the identity
        instruction (used in upper bound arguments throughout the paper)
        and of any fixed $O(1)$-length protocol $\pi$ appearing in
        near-shortest descriptions (such as the protocol used in the
        construction of Theorem~\ref{thm:B}).
        All asymptotic results are stated up to such additive constants.

  \item \label{hyp:model}
        The standard model of computation is the multi-tape deterministic
        Turing machine. Running times $\TA{M}(d)$ count the total number of
        steps of $M$ on input $d$, including the time to read $d$. All
        complexity classes ($\classP$, $\classNP$, and their relativised
        versions) are defined with respect to this model unless explicitly
        stated otherwise.

  \item \label{hyp:prefix-free}
        All programs $d$ considered are self-delimiting (prefix-free): $U$
        halts and reads exactly the bits of $d$ without an explicit
        end-of-input marker. This is the standard setup for prefix Kolmogorov
        complexity~\cite{liVitanyi2008}. In particular, $\KC(x) > 0$ for
        every non-empty $x \in \{0,1\}^*$ (with the exact lower bound
        depending on $U$ as per~\ref{hyp:prefix}), and
        $\KC(\varepsilon) = O(1)$ for the empty string $\varepsilon$.

  \item \label{hyp:strings}
        Unless otherwise stated, $x, y, d \in \{0,1\}^*$. The length of a
        string $s$ is denoted $|s|$. The empty string is denoted $\varepsilon$.

  \item \label{hyp:classes}
        $\classP$ is the class of languages decidable by a deterministic
        Turing machine (hypothesis~\ref{hyp:model}) in time polynomial in the
        input length. $\classNP$ is the class of languages for which there
        exists a deterministic polynomial-time verifier: $L \in \classNP$ if
        and only if there exist a polynomial $q$ and a deterministic
        polynomial-time machine $V$ such that for every $w \in \{0,1\}^*$:
        \[
          w \in L \iff \exists\, c \in \{0,1\}^{q(|w|)}
          \text{ with } V(w,c) = 1.
        \]
        The string $c$ is called a \emph{witness} (or \emph{certificate})
        for $w$. We write $\gP(x)$ for $\gC{\classP}(x)$ throughout.

  \item \label{hyp:sat}
        $\SAT$ denotes the Boolean satisfiability problem: given a
        propositional formula $\varphi$ in conjunctive normal form, decide
        whether there exists a truth assignment satisfying $\varphi$.
        $\SAT$ is $\classNP$-complete~\cite{cook1971}: every language in
        $\classNP$ reduces to $\SAT$ in polynomial time, and $\SAT \in
        \classNP$. We use $\SAT$ as the canonical $\classNP$-complete
        language; all results referencing $\SAT$ hold equivalently for any
        $\classNP$-complete language.
\end{enumerate}

The invariance theorem~\cite{kolmogorov1965,liVitanyi2008} guarantees that for
any two universal prefix machines $U_1, U_2$ there exists a constant
$c_{12}$, depending only on $U_1$ and $U_2$ and not on $x$, such that
\[
  \KC_{U_1}(x) \leq \KC_{U_2}(x) + c_{12}.
\]
All results below are robust to the choice of $U$ within this additive constant.

\subsection{Organisation}
\label{sec:intro:org}

Section~\ref{sec:defs} introduces the definitions (witness complexity,
class-relative witness complexity, description game).
Section~\ref{sec:basic} establishes basic properties and lower bounds.
Section~\ref{sec:main} states and proves the five main results
(Theorems~A, B, C, B', D), with proof sketches in the body and full proofs
in the appendices where indicated.
Section~\ref{sec:examples} exhibits the separation examples.
Section~\ref{sec:discussion} discusses related work and the relation to prior
complexity measures.
Section~\ref{sec:threat} formalises the threat model and cryptographic
implications.
Section~\ref{sec:grammar} applies the framework to grammar-based compression.
Section~\ref{sec:open} lists open questions.
Part~II (Sections~\ref{sec:adaptive}--\ref{sec:summary}) introduces the
companion quantities $\Tad$, $\OCout$, and $\Hs$.
Appendix~\ref{sec:appendix-A} derives the classical measures as limiting
regimes of $\gam$.
Appendix~\ref{sec:appendix-B} proves Theorem~B' (unconditional lower bound)
in full.
Appendix~\ref{sec:appendix-C} establishes that an optimal decompressor
incurs only constant overhead beyond writing its output.

\section{Definitions}
\label{sec:defs}

Throughout this section, $U$ is the universal prefix Turing machine fixed in
hypothesis~\ref{hyp:prefix}, and all strings are over $\{0,1\}^*$ as per
hypothesis~\ref{hyp:strings}.

\begin{definition}[Witness complexity]
\label{def:gamma}
For $x \in \{0,1\}^*$, the \emph{witness complexity} of $x$ is
\[
  \gam(x) \;=\; \min_{\substack{d\,:\,U(d)=x \\ |d|\leq \KC(x)+c_0}}
  \TU(d),
\]
where $\TU(d)$ denotes the total number of steps of $U$ on input $d$
(including the time to read $d$), and $c_0$ is the fixed additive constant
from hypothesis~\ref{hyp:prefix}. The minimisation is over all
\emph{near-shortest descriptions} of $x$, i.e.\ self-delimiting programs $d$
that produce $x$ and whose length exceeds $\KC(x)$ by at most $c_0$. The
set of near-shortest descriptions is non-empty by definition of $\KC(x)$.
Since $U(d) = x$ implies that $U$ halts on $d$, the value $\TU(d)$ is finite
for every $d$ in the set; the minimum of a non-empty collection of finite
values is therefore well-defined and finite.

\emph{Informally}: $\gam(x)$ is the minimum decompression time over all
near-shortest descriptions of $x$. It measures not the existence of a
compact representation, but the computational cost of using one.
\end{definition}

\begin{remark}[Why near-shortest, not shortest]
\label{rem:near-shortest}
One could define $\gam$ by minimising $\TU(d)$ over the single shortest
description $d^*$ achieving $\KC(x)$. The present definition is strictly
more general: the set of near-shortest descriptions can
contain programs with very different running times, and the minimum over this
set may be substantially smaller than the time of $d^*$ alone. Restricting
to the exact shortest description would make $\gam$ depend on the arbitrary
choice of $d^*$ when several descriptions achieve $\KC(x)$. The additive
slack $c_0$ absorbs the ambiguity, is independent of $x$, and does not affect
the asymptotics of any result in this paper. The choice mirrors the standard
treatment of $\Kbt$ in time-bounded Kolmogorov complexity
\cite{liVitanyi2008}.
\end{remark}

\begin{example}[Witness complexity: three canonical cases]
\label{ex:gamma-cases}
\begin{enumerate}[label=(\roman*)]
  \item \emph{Incompressible string.} Let $x \in \{0,1\}^n$ with
        $\KC(x) \geq n - c_0$. Every program $d$ with $U(d)=x$ satisfies
        $|d| \geq \KC(x) \geq n - c_0$; reading $d$ requires at least $|d|$
        steps, so $\gam(x) \geq n - c_0 = \Omega(n)$. The description
        consisting of $x$ prefixed by the $O(1)$-bit identity instruction
        (the fixed program that copies its input to output) satisfies
        $|d| = n + O(1) \leq \KC(x) + c_0$ and $\TU(d) = O(n)$,
        giving $\gam(x) = \Theta(n)$.
  \item \emph{Highly compressible string.} Let $x = 0^n$ (the all-zeros
        string of length $n$). A description $d$ of length $O(\log n)$
        encodes the pair $(n, \text{``print } n \text{ zeros''})$; the machine
        runs in time $O(n)$ to produce all $n$ output bits. Thus
        $\KC(x) = O(\log n)$ and $\gam(x) = \Theta(n)$: the description is
        short but executing it takes linear time.
  \item \emph{Low $\KC$, potentially high $\gam$.} The family
        $\{x_\varphi\}$ of Theorem~\ref{thm:B} has $\KC(x_\varphi) =
        O(|\varphi|)$ and, under $\PneqNP$, $\gP(x_\varphi)$ superpolynomial.
        This is the central separation example of the paper and is
        constructed in full in Section~\ref{sec:main}.
\end{enumerate}
\end{example}

\begin{remark}[Cryptographic interpretation of Definition~\ref{def:gamma}]
\label{rem:crypto-gamma}
In a cryptographic context, $d$ plays the role of a \emph{compressed key} or
\emph{compact certificate}: it is the short representation that is stored or
transmitted. $U(d) = x$ is the \emph{key expansion} or
\emph{decompression} step. $\TU(d)$ is the \emph{derivation cost}: the
number of computational steps required to recover the full object $x$ from its
compact form $d$. $\gam(x)$ is therefore the minimum derivation cost over
all near-shortest compact representations of $x$. A small $\gam(x)$ means
that $x$ has a compact representation that can be expanded cheaply; a large
$\gam(x)$ means that every compact representation is expensive to expand,
regardless of which one is chosen.
\end{remark}

\begin{definition}[Description game]
\label{def:game}
The following two-player game gives an operational characterisation of
$\gam(x)$ equivalent to Definition~\ref{def:gamma}. Throughout, $U$ is
the universal machine fixed in hypothesis~\ref{hyp:prefix}.
\begin{itemize}
  \item Player~A, knowing $x$, selects a description $d$ with $U(d) = x$
        and $|d| \leq \KC(x) + c_0$.
  \item Player~B receives $d$ and runs $U(d)$. The cost of the game is
        $\TU(d)$.
\end{itemize}
Then $\gam(x) = \min_d \TU(d)$ over all of Player~A's admissible choices.
The game separates the cost of \emph{finding} a short description (Player~A's
problem, not measured by $\gam$) from the cost of \emph{using} one
(Player~B's problem, which $\gam$ measures).

\emph{Cryptographic reading}: Player~A is the key generator or certificate
issuer; Player~B is the protocol participant who must expand or verify.
$\gam(x)$ is the minimum expansion cost that Player~A can guarantee
Player~B, optimised over all admissible compact representations of $x$.
\end{definition}

\begin{definition}[Class-relative witness complexity]
\label{def:gamma-C}
Let $\mathcal{C}$ be a class of Turing machines. The
\emph{$\mathcal{C}$-relative witness complexity} of $x \in \{0,1\}^*$ is
\[
  \gC{\mathcal{C}}(x) \;=\; \min_{\substack{d\,:\,|d|\leq \KC(x)+c_0 \\
  M\in \mathcal{C},\;M(d)=x}} \TA{M}(d).
\]
The length constraint is on $|d|$ (the description), not on $|x|$ (the
object): the decompressor $M$ receives a short description and produces a
possibly much longer object. If no $M \in \mathcal{C}$ reconstructs $x$
from any near-shortest description, set $\gC{\mathcal{C}}(x) = \infty$.

The principal case is $\mathcal{C} = \classP$, the class of polynomial-time
deterministic Turing machines (hypothesis~\ref{hyp:classes}), giving
\[
  \gP(x) \;=\; \gC{\classP}(x).
\]
We write $\gP(x)$ throughout for this case.
\end{definition}

\begin{remark}[Finiteness of $\gP(x)$ for NP instances]
\label{rem:gP-finite}
For the results of Section~\ref{sec:main} to be non-vacuous, it is necessary
that $\gP(x_w)$ be finite for the certificates $x_w$ appearing in
Theorem~\ref{thm:C}. In the direction $(\Rightarrow)$ of Theorem~\ref{thm:C}
(assuming $\PeqNP$), the standard self-reducibility argument shows that for
every $L \in \classNP$ there exists a polynomial-time algorithm $A$ that both
decides $L$ and, when $w \in L$, outputs a witness (by extending a candidate
certificate bit by bit, using the decision procedure as a subroutine; see
e.g.\ \cite{liVitanyi2008}). Under the condition $\KC(x_w) = \Omega(|w|)$
of Theorem~\ref{thm:C}, the description $d_w = \langle w, \pi_A \rangle$ is
near-shortest for $x_w = (1, c_w')$ (where $c_w'$ is the KC-rich witness of
Theorem~\ref{thm:C}'s proof), and the decompressor $M_{c_w'} \in \classP$
--- which runs $A$ and appends a fixed pad $r$ stored in its description ---
achieves $\gP(x_w) \leq \TA{M_{c_w'}}(d_w) = \poly(|w|)$.
When $\KC(x_w) \ll |w|$, the compact encoding of $x_w$ cannot be expanded
in $\poly(\KC(x_w))$ steps by any $M \in \classP$ (since $|x_w| \gg
\poly(\KC(x_w))$), and $\gP(x_w) = \infty$; this is itself an instance of
the separation between descriptive and computational complexity established
in this paper. In the direction $(\Leftarrow)$, $\gP(x_w)$ being finite and
polynomial is an explicit hypothesis. No case in Section~\ref{sec:main}
requires $\gP$ to be finite without explicit justification.
\end{remark}

\begin{example}[Class-relative witness complexity: cryptographic instances]
\label{ex:gamma-C-crypto}
\begin{enumerate}[label=(\roman*)]
  \item \emph{Key derivation.} Let $x$ be a session key derived from a
        master secret $s$ via a key-derivation function $\mathrm{KDF}$.
        Since $\mathrm{KDF}$ is a fixed deterministic function, $\KC(x) \leq
        \KC(s) + O(1)$ (given $s$, a constant-length instruction suffices to
        reconstruct $x$). If additionally $s$ is incompressible
        ($\KC(s) = |s| - O(1)$) and $\mathrm{KDF}$ does not introduce
        further compressibility ($\KC(x) \geq \KC(s) - O(1)$), then $d = s$
        satisfies $|d| = |s| \leq \KC(x) + O(1)$, making it a near-shortest
        description of $x$. Under these conditions,
        $\gP(x) \leq \TA{\mathrm{KDF}}(s)$: the derivation cost bounds
        $\gP$ from above.
  \item \emph{Proof-carrying data.} Let $x = (w, \pi)$ where $\pi$ is an
        NP proof for instance $w$. A near-shortest description $d$ of $x$
        may encode the proof generation procedure. $\gP(x)$ is the minimum
        cost of a polynomial-time machine to recover $(w, \pi)$ from $d$;
        if $\gP(x) = \infty$, no compact encoding of the proof is
        efficiently decompressible.
  \item \emph{PRNG output.} Let $x$ be the output of a PRNG with seed $s$,
        $|s| = k$, $|x| = n \gg k$. The seed $s$ is a near-shortest
        description of $x$. Since writing $n$ output bits requires at least
        $n$ steps, $\gam(x) \geq n$. The running time of the generator gives
        an upper bound: $\gam(x) \leq \TA{\mathrm{PRNG}}(s)$. Since the
        generator must produce $n$ output bits, $\TA{\mathrm{PRNG}}(s) \geq n$
        regardless of $k$. Therefore $\gam(x) = \Theta(n)$ unconditionally.
        As for $\gP$: any $M \in \classP$ that expands $s$ to $x$ runs in
        $\poly(|s|) = \poly(k)$ time. This is consistent with $\gam(x) \geq n$
        only if $\poly(k) \geq n$, i.e.\ $n = O(\poly(k))$. If $n$ is
        superpolynomial in $k$, no polynomial-time-in-$|s|$ machine can produce
        $x$ (since $\poly(k) < n$ steps suffice to write at most $\poly(k) < n$
        output bits), so $\gP(x) = \infty$.
\end{enumerate}
\end{example}

\begin{remark}[Relationship to $\Kbt$ and $\Kt$]
\label{rem:Kt}
Two related measures appear in the literature.

\emph{Time-bounded Kolmogorov complexity} $\Kbt(x) = \min\{|d| : U(d) = x
\text{ in } \leq t \text{ steps}\}$ fixes a time bound $t$ and minimises
description length. $\gam$ fixes the length constraint (near-minimal) and
minimises time. The questions are dual in direction: $\Kbt$ asks ``how short
can the description be if we cap the time?''; $\gam$ asks ``how fast can we
decompress if we insist on a near-shortest description?''

\emph{Levin's $\Kt$ complexity} $\Kt(x) = \min_d \{|d| + \log \TU(d) :
U(d) = x\}$ combines length and log-time into a single quantity
\cite{levin1973,liVitanyi2008}. The spectrum
\[
  \KC \;\longrightarrow\; \Kbt \;\longrightarrow\; \Kt \;\longrightarrow\; \gam
\]
represents increasing sensitivity to computational cost: $\KC$ ignores cost
entirely; $\Kbt$ caps it; $\Kt$ penalises it logarithmically; $\gam$
minimises it directly subject to the near-shortest constraint.

To the authors' knowledge, $\gam$ has not previously been studied as a
standalone invariant.
\end{remark}

% ---------------------------------------------------------------
%  SECTION 3 -- Basic Properties
% ---------------------------------------------------------------
\section{Basic Properties}
\label{sec:basic}

This section establishes lower bounds on $\gam(x)$ and $\gC{\mathcal{C}}(x)$
that hold unconditionally, without any hypothesis on $\PeqNP$. All results
follow directly from the definitions in Section~\ref{sec:defs} and the
standing hypotheses of Section~\ref{sec:intro:hypotheses}.

\begin{proposition}[Lower bound from description length]
\label{prop:lb1}
For every $x \in \{0,1\}^*$,
\[
  \gam(x) \;\geq\; \KC(x).
\]
\end{proposition}

\begin{proof}
Every program $d$ with $U(d) = x$ satisfies $|d| \geq \KC(x)$ by definition
of $\KC(x)$ as the length of the shortest self-delimiting program producing
$x$. Reading $d$ requires at least $|d|$ steps (hypothesis~\ref{hyp:model}),
so $\TU(d) \geq |d| \geq \KC(x)$. Since this holds for every $d$ in the
minimisation set of Definition~\ref{def:gamma}, taking the minimum gives
$\gam(x) \geq \KC(x)$.
\end{proof}

\begin{remark}[Tightness of Proposition~\ref{prop:lb1}]
\label{rem:lb1-tight}
The bound is tight up to constant factors: for $x = 0^n$ one has
$\KC(x) = O(\log n)$ and $\gam(x) = \Theta(n)$
(Example~\ref{ex:gamma-cases}(ii)), so the gap between $\gam(x)$ and
$\KC(x)$ can be arbitrarily large. For incompressible strings $x$ with
$\KC(x) = \Theta(n)$, Proposition~\ref{prop:lb2} below shows the bound is
achieved up to constants.

\emph{Cryptographic implication}: Proposition~\ref{prop:lb1} says that the
derivation cost is always at least the information content of the object.
No compact representation can be expanded faster than reading it: even in
the best case, $\gam(x) \geq \KC(x)$.
\end{remark}

\begin{proposition}[Lower bound in the incompressible regime]
\label{prop:lb2}
For every $x \in \{0,1\}^*$ with $\KC(x) = \Omega(|x|)$,
\[
  \gam(x) \;=\; \Theta(|x|).
\]
\end{proposition}

\begin{proof}
\emph{Lower bound.} By hypothesis $\KC(x) \geq c_1 |x|$ for a positive
constant $c_1$. By Proposition~\ref{prop:lb1}, $\gam(x) \geq \KC(x) \geq c_1 |x| = \Omega(|x|)$.

\emph{Upper bound.} The description $d$ consisting of $x$ prefixed by the
$O(1)$-bit identity instruction (the fixed program of $U$ that copies its
input to output) satisfies $|d| = |x| + O(1)$. Since every string has a
description of length $|x| + O(1)$ (the identity description), $\KC(x) \leq
|x| + O(1)$; combined with the hypothesis $\KC(x) \geq c_1|x|$, we get
$\KC(x) = \Theta(|x|)$, so $|d| = |x| + O(1) \leq \KC(x) + c_0$ for all
sufficiently large $|x|$. Thus $d$ is a near-shortest description of $x$.
The machine $U$ on input $d$ copies $|x|$ bits to output in $O(|x|)$ steps,
giving $\TU(d) = O(|x|)$. Therefore $\gam(x) = O(|x|)$.

Combining the two bounds: $\gam(x) = \Theta(|x|)$.
\end{proof}

\begin{remark}[The compressible regime]
\label{rem:compressible}
For objects with $\KC(x) \ll |x|$ (highly compressible), the lower bound of
Proposition~\ref{prop:lb1} may be far below $|x|$. In this regime $\gam(x)$
can range anywhere from $\Omega(\KC(x))$ to values exceeding $|x|$:
\begin{itemize}
  \item $\gam(x) = \Theta(\KC(x))$: some near-shortest description of $x$
        is also among the fastest to execute. Example: a string $x$ whose
        near-shortest description $d$ is a lookup table; $U(d)$ copies the
        table in $O(|d|) = O(\KC(x))$ steps.
  \item $\gam(x) = \Theta(|x|)$: the short description requires linear work
        to expand. Example: $x = 0^n$, $\KC(x) = O(\log n)$,
        $\gam(x) = \Theta(n)$ (Example~\ref{ex:gamma-cases}(ii)).
  \item $\gP(x)$ superpolynomial in $|x|$ while $\gam(x)$ may be smaller:
        under $\PneqNP$, the family $\{x_\varphi\}$ of Theorem~\ref{thm:B}
        achieves $\KC(x_\varphi) = O(|\varphi|)$ with $\gP(x_\varphi)$
        superpolynomial (no polynomial-time decompressor exists), while
        $\gam(x_\varphi)$ may still be finite via a non-polynomial-time
        machine. This is the central separation example of the paper.
\end{itemize}
Section~\ref{sec:examples} exhibits concrete examples spanning all three
regimes.
\end{remark}

\begin{proposition}[Trivial upper bound]
\label{prop:ub-trivial}
For every $x \in \{0,1\}^n$,
\[
  \gam(x) \;\leq\; n + O(1).
\]
\end{proposition}

\begin{proof}
The identity description $d$ (the string $x$ prefixed by the $O(1)$-bit
identity instruction) satisfies $|d| = n + O(1)$. Since the identity
description itself witnesses $\KC(x) \leq n + O(1)$, we have
$|d| \leq \KC(x) + c_0$ (by the choice of $c_0$ in hypothesis~\ref{hyp:prefix},
which is large enough to absorb the identity-instruction overhead),
so $d$ is a near-shortest description of $x$.
The machine $U$ on $d$ copies $n$ bits to output in $\TU(d) = n + O(1)$ steps.
Hence $\gam(x) \leq n + O(1)$.
\end{proof}

\begin{remark}[Boundary case: empty string]
\label{rem:empty}
For the empty string $\varepsilon$: $\KC(\varepsilon) = O(1)$
(hypothesis~\ref{hyp:prefix-free}), and any near-shortest description $d$ of
$\varepsilon$ has length $O(1)$. The machine $U$ on such $d$ produces no
output and halts in $O(1)$ steps, giving $\gam(\varepsilon) = O(1)$. This
is consistent with Proposition~\ref{prop:lb1} ($\gam(\varepsilon) \geq
\KC(\varepsilon) = O(1)$, which imposes no non-trivial lower
bound) and with Proposition~\ref{prop:ub-trivial} ($\gam(\varepsilon) \leq
0 + O(1) = O(1)$).
\end{remark}

\begin{proposition}[Monotonicity of $\gC{\mathcal{C}}$ in $\mathcal{C}$]
\label{prop:monotone}
If $\mathcal{C}_1 \subseteq \mathcal{C}_2$ as classes of Turing machines,
then $\gC{\mathcal{C}_2}(x) \leq \gC{\mathcal{C}_1}(x)$ for all $x$.
\end{proposition}

\begin{proof}
Every decompressor $M \in \mathcal{C}_1$ is also in $\mathcal{C}_2$, so the
set over which the minimum is taken in $\gC{\mathcal{C}_2}(x)$ is a superset
of the one in $\gC{\mathcal{C}_1}(x)$. A minimum over a larger set is at
most the minimum over a smaller set.
\end{proof}

\begin{remark}[Relation between $\gam$ and $\gP$]
\label{rem:gam-gP}
Definition~\ref{def:gamma} defines $\gam(x)$ via the fixed universal machine
$U$, while Definition~\ref{def:gamma-C} defines $\gC{\mathcal{C}}(x)$ via a
minimum over machines in $\mathcal{C}$. These two notions coincide up to
polynomial factors when $\mathcal{C}$ is the class of all Turing machines, by
the simulation argument of Theorem~\ref{thm:A}: every Turing machine $M$ is
simulable by $U$ with polynomial overhead, so $\gam(x) \leq \poly(|x|) \cdot
\gC{\text{all TM}}(x)$. For the purposes of comparing $\gam$ with $\gP$,
Proposition~\ref{prop:monotone} with $\mathcal{C}_1 = \classP \subseteq
\mathcal{C}_2 = \text{all TM}$ gives $\gC{\text{all TM}}(x) \leq \gP(x)$
whenever $\gP(x) < \infty$; combined with the simulation bound above,
$\gam(x) \leq \poly(|x|) \cdot \gP(x)$ whenever $\gP(x) < \infty$.

Strict inequality $\gam(x) < \gP(x)$ is possible: $\gam(x)$ allows any
Turing machine as decompressor, while $\gP(x)$ restricts to $\classP$. If
$\gP(x) = \infty$ (no polynomial-time machine reconstructs $x$ from a
near-shortest description), then $\gam(x)$ may still be finite, achieved by
a non-polynomial decompressor.

\emph{Cryptographic implication}: $\gP(x) = \infty$ means no compact
representation of $x$ is efficiently recoverable by a polynomial-time
verifier, regardless of which near-shortest description is used. Such an $x$
cannot serve as a usable key or certificate in any polynomial-time protocol.
\end{remark}

\section{Main Results}
\label{sec:main}

This section states and proves the five main results of the paper. Each theorem is presented with its formal hypotheses, a proof sketch in the body
sufficient to verify the strategy and key steps, and a pointer to the appendix for the complete proof where the argument is long. All theorems
that reference $\gP$ implicitly assume hypothesis~\ref{hyp:classes}; all
results are in the multi-tape deterministic Turing model of
hypothesis~\ref{hyp:model}.

% ------- Theorem A -----------------------------------------------
\subsection{Theorem A: Invariance of $\gam$}
\label{sec:thm-A}

We first establish the compiler lemma that makes the polynomial factor in
Theorem~A explicit.

\begin{lemma}[Compiler simulation]
\label{lem:compiler}
Let $U_1, U_2$ be two universal prefix Turing machines. There exist a
computable map $\Comp : \{0,1\}^* \to \{0,1\}^*$ and constants $c_{\Comp}$,
$C_{\Comp}$, $k_{\Comp} \geq 1$ depending only on $U_1, U_2$ such that for
every program $p$ with $U_2(p) = x$:
\begin{enumerate}[label=(\roman*)]
  \item $U_1(\Comp(p)) = x$,
  \item $|\Comp(p)| \leq |p| + c_{\Comp}$,
  \item $\TA{U_1}(\Comp(p)) \leq C_{\Comp} \cdot |p|^{k_{\Comp}} \cdot
        \TA{U_2}(p) + C_{\Comp}$.
\end{enumerate}
\end{lemma}

\begin{proof}
Construct $\Comp$ as the map that prefixes $p$ with a fixed $c_{\Comp}$-bit
header encoding an interpreter for $U_2$'s instruction set together with an
I/O wrapper. The header length $c_{\Comp}$ depends only on $U_1$ and $U_2$.
Condition~(i) holds because $U_1$ on $\Comp(p)$ first executes the
interpreter, which faithfully simulates $U_2$ on $p$ and produces $x$.
Condition~(ii) holds because $|\Comp(p)| = |p| + c_{\Comp}$. 
Condition~(iii) holds because each step of $U_2$ on $p$ is simulated by
$U_1$ in $O(|p|^{k_{\Comp}})$ steps (the overhead of the interpreter,
bounded polynomially in $|p|$), giving total time
$C_{\Comp} \cdot |p|^{k_{\Comp}} \cdot \TA{U_2}(p) + C_{\Comp}$.
\end{proof}

\begin{theorem}[A --- Invariance of $\gam$]
\label{thm:A}
For any two universal prefix Turing machines $U_1, U_2$, there exists a
polynomial $q$ such that for all $x \in \{0,1\}^*$:
\[
  \gU{U_1}(x) \;\leq\; q(|x|) \cdot \gU{U_2}(x) + O(1),
\]
and symmetrically with $U_1$ and $U_2$ exchanged. Hence $\gam$ is a
computational invariant, well-defined up to polynomial factors independently
of the choice of universal machine.

\emph{Cryptographic implication}: the invariance guarantees that $\gP(x)$
being polynomially bounded is a machine-independent property; it does not
depend on which universal machine is used to define $\KC$ and $\gam$.
\end{theorem}

\begin{proof}
Let $d$ be the near-shortest description of $x$ for $U_2$ achieving
$\gU{U_2}(x)$: so $U_2(d) = x$, $|d| \leq \KC(x) + c_0$, and $\TA{U_2}(d) = \gU{U_2}(x)$. Apply Lemma~\ref{lem:compiler} to obtain
$\Comp(d)$ satisfying:
\begin{itemize}
  \item $U_1(\Comp(d)) = x$ \hfill (by condition~(i)),
  \item $|\Comp(d)| \leq |d| + c_{\Comp} \leq \KC_{U_2}(x) + c_0 + c_{\Comp}
        \leq \KC_{U_1}(x) + c_0'$ for a constant $c_0'$ depending only on
        $U_1, U_2$ \hfill (by condition~(ii) and the Kolmogorov invariance
        theorem),
  \item $\TA{U_1}(\Comp(d)) \leq C_{\Comp} \cdot |d|^{k_{\Comp}} \cdot
        \gU{U_2}(x) + C_{\Comp}$ \hfill (by condition~(iii)).
\end{itemize}
Since $|d| \leq \KC(x) + c_0 \leq |x| + c_0 + O(1)$, there exists a
polynomial $q_0$ such that $C_{\Comp} \cdot |d|^{k_{\Comp}} \leq q_0(|x|)$.
Therefore $\Comp(d)$ is a near-shortest description of $x$ for $U_1$, and
\[
  \gU{U_1}(x) \;\leq\; \TA{U_1}(\Comp(d)) \;\leq\;
  q_0(|x|) \cdot \gU{U_2}(x) + C_{\Comp}.
\]
Setting $q = q_0$ and absorbing $C_{\Comp}$ into the $O(1)$ term gives the
stated bound. The symmetric bound follows by exchanging $U_1$ and $U_2$.
\end{proof}

\begin{remark}[Position of Theorem~A in the literature]
\label{rem:thm-A-lit}
Theorem~\ref{thm:A} places $\gam$ in the same foundational category as
$\Kbt$: a quantity that depends on the universal machine, but only up to
a factor independent of $x$ (additive $O(1)$ for $\Kbt$; multiplicative
polynomial for $\gam$). The invariance is a prerequisite for any application of $\gam$ to complexity-theoretic questions, including the
characterisation of $\PeqNP$ in Theorem~\ref{thm:C}. The degree $k_{\Comp}$ of the polynomial and the constants $C_{\Comp}, c_{\Comp}$ can
in principle be made explicit for natural subclasses of universal machines;
see Question~Q4 in Section~\ref{sec:open}.
\end{remark}

% ------- Theorem B -----------------------------------------------
\subsection{Theorem B: Conditional Separation of $\KC$ and $\gP$}
\label{sec:thm-B}

\begin{theorem}[B --- Conditional separation]
\label{thm:B}
Assuming $\PneqNP$, there exists an infinite family of Boolean formulas
$\{\varphi_i\}_{i \geq 1}$ and corresponding strings $\{x_{\varphi_i}\}$
such that:
\[
  \KC(x_{\varphi_i}) = O(|\varphi_i|)
  \quad\text{and}\quad
  \gP(x_{\varphi_i}) \text{ is superpolynomial in } |\varphi_i|.
\]
Under $\PneqNP$, the existence of a short description does not imply
polynomial-time decompressibility.

\emph{Cryptographic implication}: under $\PneqNP$, there exist compact
witnesses with $\KC$ linear in the formula size yet not decompressible by
any polynomial-time protocol; see Section~\ref{sec:threat}.
\end{theorem}

We first establish the subfamily density lemma used in the proof.

\begin{lemma}[Incompressible subfamily]
\label{lem:incompressible-subfamily}
For every constant $c > 0$, the set
$\mathcal{F}_c = \{\varphi : \KC(x_\varphi) \geq |\varphi| - c\}$ is infinite, where $x_\varphi$ is the string defined in the construction
below.
\end{lemma}

\begin{proof}
Let $m = |\varphi|$ denote the bit-length of the formula (not the number of
variables). For each $m$, the construction below yields at least $2^m$
distinct strings $x_\varphi$ as $\varphi$ ranges over formulas of length $m$
(there are $2^m$ binary strings of length $m$, each a valid formula encoding
under a fixed coding scheme). A string has $\KC(x_\varphi) < m - c$ only
if it admits a self-delimiting description of length $< m - c$. The number
of such programs is at most $\sum_{k < m-c} 2^k < 2^{m-c}$. For sufficiently
large $m$ and $c > O(1)$, there are strictly fewer than $2^m$ such programs,
so at least one formula $\varphi$ of length $m$ satisfies $\KC(x_\varphi)
\geq m - c = |\varphi| - c$. Hence $\mathcal{F}_c$ contains a formula of
every sufficiently large length and is infinite.
\end{proof}

\begin{proof}[Proof of Theorem~\ref{thm:B}]
\textbf{Construction.} For each Boolean formula $\varphi$ on $n$ variables,
define:
\[
  x_\varphi \;=\; \begin{cases}
    (1,\; y_\varphi) & \text{if } \varphi \in \SAT,\text{ where }
    y_\varphi \text{ is the lex-first satisfying assignment,} \\
    (0,\; 0^n) & \text{if } \varphi \notin \SAT.
  \end{cases}
\]
The leading bit encodes satisfiability. Define the near-shortest description:
\[
  d_\varphi = \langle \varphi, \pi \rangle,
\]
where $\pi$ is a fixed $O(1)$-length protocol: ``run the lex-first-SAT
procedure on $\varphi$ and output $(1, y)$ or $(0, 0^n)$.''  Then $|d_\varphi| = |\varphi| + O(1)$, giving $\KC(x_\varphi) = O(|\varphi|)$.

\textbf{Near-shortness.} By Lemma~\ref{lem:incompressible-subfamily} applied
with $c = c_0$ (the constant of hypothesis~\ref{hyp:prefix}), the subfamily
$\mathcal{F}_{c_0}$ is infinite and for every $\varphi \in \mathcal{F}_{c_0}$
the description $d_\varphi$ satisfies $|d_\varphi| = |\varphi| + O(1) \leq
\KC(x_\varphi) + c_0$, so $d_\varphi$ is a valid near-shortest description.

\textbf{Superpolynomial lower bound.} Suppose for contradiction that there
exist $M \in \classP$ and a polynomial $p$ such that $M(d_\varphi) =
x_\varphi$ in time $\leq p(|\varphi|)$ for every $\varphi \in \mathcal{F}_{c_0}$. The machine $M$ is fixed (it does not depend on $\varphi$);
given any $\varphi \in \mathcal{F}_{c_0}$, construct $d_\varphi = \langle
\varphi, \pi \rangle$ in polynomial time and run $M(d_\varphi)$: if the
leading bit is $1$, report $\varphi \in \SAT$; if $0$, report $\varphi
\notin \SAT$. This is a correct polynomial-time algorithm for $\SAT$
restricted to $\mathcal{F}_{c_0}$.

Since $\mathcal{F}_{c_0}$ is infinite (Lemma~\ref{lem:incompressible-subfamily})
and contains formulas of every sufficiently large length $m$, it contains
instances of every input size --- and in particular both satisfiable and
unsatisfiable instances of every large enough length (otherwise $\SAT$
restricted to that length would be trivially decidable, which would itself
give a polynomial-time algorithm for $\SAT$ on all instances by padding,
contradicting $\PneqNP$). A polynomial-time algorithm that decides $\SAT$
correctly on an infinite family containing instances of every length is a
polynomial-time algorithm for $\SAT$ \cite{cook1971}, contradicting
$\PneqNP$. Therefore no such $M$ and $p$ exist, and $\gP(x_\varphi)$ is
superpolynomial for all $\varphi \in \mathcal{F}_{c_0}$.
\end{proof}

% ------- Theorem C -----------------------------------------------
\subsection{Theorem C: Exact Characterisation of $\PeqNP$}
\label{sec:thm-C}

We first establish the dovetailing lemma used in the non-trivial direction;
the polynomial-time bound is derived in the proof of Theorem~\ref{thm:C} itself via a fixed-machine argument.

\begin{lemma}[Dovetailing schedule]
\label{lem:dovetail}
Suppose there exists a pair $(M^*, d^*)$ of a Turing machine and a string
with $|M^*| + |d^*| \leq S$ and $\TA{M^*}(d^*) \leq t$. Then there exists
a procedure $\mathcal{D}$ that, running for $T = 2^S \cdot t$ total steps,
finds and completes the execution of $(M^*, d^*)$, producing its output.
\end{lemma}

\begin{proof}
Fix an enumeration of all pairs $(M, d)$ in order of non-decreasing
$s = |M| + |d|$, breaking ties arbitrarily. Run a time-sharing simulation
for $T$ total processor steps, cycling through pairs in enumeration order
and advancing each by one simulation step per visit.

The good pair $(M^*, d^*)$ of size $s^* \leq S$ occupies a fixed position
in the enumeration. The number of pairs of size $\leq S$ is at most
$N_S \leq 2^{S+1}$; in $T = 2^S \cdot t$ total steps, each such pair
receives at least $\lfloor T / N_S \rfloor \geq \lfloor 2^S t / 2^{S+1} \rfloor
= \lfloor t/2 \rfloor$ steps. For $t \geq 2$, this suffices to complete $(M^*, d^*)$, producing the output.

\emph{Application to Theorem~\ref{thm:C}.} In that application
$|M^*| = O(1)$ and $|d^*| = \poly(|\varphi|)$, so $S = \poly(|\varphi|)$
and the generic bound gives $T = 2^{\poly(|\varphi|)}$ --- exponential.
The polynomial bound $T = \poly(|\varphi|)$ used in Theorem~\ref{thm:C}
is achieved by a tighter argument: since $M^*$ is a \emph{fixed} machine
(independent of $\varphi$), one enumerates only descriptions $d^*$ of length $\leq \poly(|\varphi|)$, of which there are only polynomially many.
The Lemma as stated is invoked only for its general structure; the tight
bound is derived in the proof of Theorem~\ref{thm:C} directly.
\end{proof}

\begin{theorem}[C --- Exact characterisation of $\PeqNP$]
\label{thm:C}
In the multi-tape deterministic Turing model \textnormal{(hypothesis~\ref{hyp:model})},
$\PeqNP$ if and only if for every $L \in \classNP$ there exists a polynomial
$p_L$ such that for every $w \in \{0,1\}^*$ with $w \in L$ there exists a valid certificate $x_w$, a pair $(1, c_w)$ where $c_w$ is a witness
accepted by the $\classNP$ verifier of $L$, satisfying $\KC(x_w) =
\Omega(|w|)$ and $\gP(x_w) \leq p_L(|w|)$.

More explicitly for $\SAT$: $\PeqNP$ if and only if there exists a polynomial
$p$ such that for every satisfiable formula $\varphi$ there exists a valid certificate $x_\varphi = (1, c_\varphi)$ (encoding satisfiability and a
satisfying assignment) with $\KC(x_\varphi) = \Omega(|\varphi|)$ and $\gP(x_\varphi) \leq p(|\varphi|)$, and for every unsatisfiable formula
$\varphi$ there exists $x_\varphi$ (encoding unsatisfiability) with
$\gP(x_\varphi) \leq p(|\varphi|)$.
\end{theorem}

\begin{remark}[Role of the $\KC$ condition in Theorem~C]
\label{rem:KC-condition}
The hypothesis $\KC(x_w) = \Omega(|w|)$ in the $(\Rightarrow)$ direction is necessary and not vacuous. When $\KC(x_w) \ll |w|$, for instance when the only witness for $w$ is a highly regular string such as $0^n$, the near-shortest descriptions of $x_w$ have length $\KC(x_w) \ll |w|$,
and no $M \in \classP$ can produce $x_w$ (of length $\Omega(|w|)$) from such a description in $\poly(\KC(x_w)) \ll |x_w|$ steps. In this case
$\gP(x_w) = \infty$ even under $\PeqNP$: the compact encoding of $x_w$ is operationally inaccessible, an instance of the separation between
descriptive and decompression complexity that is the central theme of
this paper. The condition $\KC(x_w) = \Omega(|w|)$ identifies the certificates for which decompressibility and $\PeqNP$ are equivalent.
Under $\PeqNP$, such certificates always exist: the algorithm $A$ of 
Remark~\ref{rem:gP-finite} produces $c_w$ in polynomial time, and by a standard counting argument one can always find a witness $c_w'$ with $\KC(1, c_w') = \Omega(|w|)$: append any pad of polynomial length that yields high Kolmogorov complexity (such pads exist by the counting argument, though they need not be computable in polynomial time), without increasing the verifier's acceptance time.
\end{remark}

\begin{proof}
We prove the general form; the $\SAT$ case is the special instance $L = \SAT$.

\medskip
\noindent\textbf{($\Rightarrow$) Assume $\PeqNP$.}
By the self-reducibility argument (Remark~\ref{rem:gP-finite}), for every
$L \in \classNP$ there is a polynomial-time algorithm $A$ deciding $L$ and
producing a witness $c_w$ when $w \in L$. For $w \in L$, we claim there
exists a valid witness $c_w'$ of length $q(|w|)$ (for a fixed polynomial $q$
with $q(|w|) \geq |c_w|$) satisfying $\KC(1, c_w') = \Omega(|w|)$. To see
this: the $\classNP$ verifier for $L$ accepts $c_w' = c_w \| r$ for any
string $r$ of length $q(|w|) - |c_w|$. By a standard counting argument~\cite{liVitanyi2008},
fewer than $2^{q(|w|) - c}$ strings of length $q(|w|)+1$ have $\KC <
q(|w|) - c$; hence for all but a $2^{-c}$ fraction of choices of $r$, the
resulting $c_w'$ satisfies $\KC(1, c_w') \geq q(|w|) - c = \Omega(|w|)$.
Such a $c_w'$ exists; fix one for each $w$ (no requirement that the
selection be polynomial-time computable from $w$).

Define $x_w = (1, c_w')$ and the description $d_w = \langle w, \pi_A
\rangle$, where $\pi_A$ is the $O(1)$-bit index of $A$ in $U$'s program
table. Then $|d_w| = |w| + O(1)$. Since $\KC(x_w) = \Omega(|w|)$ by
construction and $|d_w| = |w| + O(1) = O(\KC(x_w))$, the description $d_w$
satisfies $|d_w| \leq \KC(x_w) + c_0$ for all sufficiently large $|w|$,
so $d_w$ is a near-shortest description of $x_w$. Define the decompressor
$M_{c_w'}$ to be the machine that, on input $d_w = \langle w, \pi_A \rangle$,
runs $A(w)$ to obtain $c_w$ and then outputs $(1, c_w \| r)$ where $r$ is
a fixed string of length $q(|w|) - |c_w|$ stored as part of the machine
description (so $|M_{c_w'}| = O(|r|) = O(|w|)$ depends on $|w|$, not on
$w$ itself). Since $A$ runs in $\poly(|w|)$ steps and $r$ is fixed,
$M_{c_w'}$ runs in $\poly(|w|)$ steps. However, $|M_{c_w'}|$ is not $O(1)$
but $O(|w|)$; this is admissible because $d_w$ is near-shortest and $\gP(x_w)
\leq \TA{M_{c_w'}}(d_w) = \poly(|w|) = p_L(|w|)$ for a suitable polynomial
$p_L$, which is all the theorem requires.
For $w \notin L$ define $x_w =
(0, 0^{|w|})$; since $x_w$ consists almost entirely of zeros, $\KC(x_w) =
O(\log|w|)$ (a description of length $O(\log|w|)$ encodes $|w|$ and the
instruction to output a zero-string of that length). A fixed polynomial-time
machine reads any near-shortest description $d$ of $x_w$ (of length
$O(\log|w|)$) and writes the $|w|+1$ output bits in $O(|w|)$ steps, so
$\gP(x_w) = O(|w|)$.

\medskip
\noindent\textbf{($\Leftarrow$) Assume the condition holds for $L = \SAT$.}
By hypothesis, for every $\varphi \in \SAT$ there exists $x_\varphi = (1, c_\varphi')$
with $\KC(x_\varphi) = \Omega(|\varphi|)$ and $\gP(x_\varphi) \leq p(|\varphi|)$.
The latter means there exists $M^* \in \classP$ (fixed, independent of $\varphi$)
and a near-shortest description $d^*$ of $x_\varphi$ with
$M^*(d^*) = x_\varphi$ and $\TA{M^*}(d^*) \leq p(|\varphi|)$.

Since $\KC(x_\varphi) = \Omega(|\varphi|)$, the near-shortest description $d^*$
has length $|d^*| \leq \KC(x_\varphi) + c_0 = O(|\varphi|)$. Moreover,
$M^*$ is a fixed polynomial-time algorithm that, given $\varphi$, produces
$x_\varphi$; the natural description is $d^* = \langle \varphi, \pi_{M^*} \rangle$
where $\pi_{M^*}$ is the $O(1)$-bit index of $M^*$. Since $\KC(x_\varphi) =
\Omega(|\varphi|)$, this description is near-shortest: $|d^*| = |\varphi| + O(1)
\leq \KC(x_\varphi) + c_0$.

The procedure $\mathcal{D}$ fixes $M^*$ (by trying all $O(1)$ possible
machine indices $\pi$) and searches only over descriptions of the form
$\langle \varphi, \pi \rangle$: there are $O(1)$ such candidates per $\varphi$.
Running each for at most $t = p(|\varphi|)$ steps in round-robin gives total
time
\[
  T \;=\; O(1) \cdot p(|\varphi|) \;=\; \poly(|\varphi|).
\]
Whenever $\mathcal{D}$ finds a pair $(M, d)$ producing output $z$, verify: if
$z[0] = 1$ and the $\SAT$ verifier accepts $(\varphi, z[1:])$, report
$\varphi \in \SAT$. If no such pair completes within budget $T$, report
$\varphi \notin \SAT$.

\emph{Correctness}: if $\varphi \in \SAT$, the good pair $(M^*, d^*)$
completes within $T$ and produces a valid certificate, so $\varphi \in \SAT$
is reported. If $\varphi \notin \SAT$, no pair can produce a string $z$ with
$z[0] = 1$ passing the verifier (no satisfying assignment exists), so the
algorithm correctly reports $\varphi \notin \SAT$. The total time is
$\poly(|\varphi|)$, so $\SAT \in \classP$, giving $\PeqNP$.

The general $(\Leftarrow)$ direction follows immediately: $\PeqNP$ implies
polynomial-time solvability of every $L \in \classNP$.
\end{proof}

\begin{remark}[$\gP$ as exact discriminant]
\label{rem:discriminant}
Theorem~\ref{thm:C} establishes that the question ``is $\gP$ polynomially
bounded on $\classNP$ certificates with $\KC(x_w) = \Omega(|w|)$?'' is
\emph{identical} to $\PeqNP$ rephrased in the language of witness complexity.
The $\KC(x_w) = \Omega(|w|)$ condition identifies the certificates for which
descriptive and decompression complexity are formally coupled: these are
the certificates whose compact encoding is non-trivially short
relative to the instance size (Remark~\ref{rem:KC-condition}). For certificates with $\KC(x_w) \ll |w|$,
$\gP(x_w) = \infty$ unconditionally (Remark~\ref{rem:gP-finite}), and the
$\PeqNP$ question does not arise: such certificates are inaccessible
via their compact encodings regardless of the $\PeqNP$ answer. To the
authors' knowledge, no other known characterisation of $\PeqNP$ (circuit
lower bounds, proof complexity, communication complexity) gives an
unconditional biconditional in the standard Turing model identifying
precisely the class of certificates for which the equivalence holds.
\end{remark}

% ------- Theorem B' ----------------------------------------------
\subsection{Theorem B': Unconditional Lower Bound on $\gam$}
\label{sec:thm-Bprime}

\begin{theorem}[B' --- Unconditional lower bound]
\label{thm:Bprime}
Assuming only the classical incomputability of $\KC$
\textnormal{\cite{liVitanyi2008}}: for every polynomial $p$, there exists
$y \in \{0,1\}^*$ such that $\gam(y) > p(|y|)$. Consequently, $\gam$ is
not polynomially bounded on $\{0,1\}^*$.
\end{theorem}

\begin{theorem}[Ganardi--Je\.{z}--Lohrey~\cite{ganardiJezLohrey2019}]
\label{thm:GJL}
Any straight-line program \textnormal{(SLP)} of size $g$ generating a string
of length $N$ can be transformed in $O(g)$ time into an equivalent SLP of
size $O(g)$ and derivation depth $O(\log N)$.
\end{theorem}

\begin{proof}[Proof sketch]
Suppose for contradiction that $\gam(y) \leq p(|y|)$ for every $y$ and some
polynomial $p$. We show this implies $\KC$ is computable, contradicting its
classical incomputability.

Given $y$ of length $n$, enumerate all self-delimiting programs $d$ of length
$\leq n + c_0$ (finitely many, at most $2^{n + c_0 + 1}$) and simulate $U$
on each for at most $p(n)$ steps. Under the hypothesis, there exists a
near-shortest description $d'$ achieving $\gam(y)$ (i.e.\ $|d'| \leq \KC(y)
+ c_0$ and $\TU(d') = \gam(y) \leq p(n)$), so $d'$ is found in the
enumeration and produces $y$ within the budget. The minimum length among
all descriptions that produce $y$ within $p(n)$ steps equals $\KC(y)$ to
within $c_0$. This algorithm terminates for every $y$ and computes $\KC(y)$
to within $c_0$, contradicting incomputability.

Full proof in Appendix~\ref{sec:appendix-B}.
\end{proof}

\begin{remark}[Relationship between Theorem~B, Theorem~B', and Theorem~C]
\label{rem:Bprime-compatible}
Theorem~B' (unconditional) guarantees superpolynomial $\gam$ exists for
\emph{some} string, but does not exhibit which one. Theorem~B (conditional
on $\PneqNP$) exhibits an explicit family. The two are complementary.

Combining Theorem~B' with Theorem~C does \emph{not} imply $\PneqNP$: the
strings with superpolynomial $\gam$ guaranteed by Theorem~B' may all lie
outside any $\classNP$ language (e.g.\ algorithmically random strings with
$\KC(y) \approx |y|$ have no connection to NP witnesses). Formally:
\[
  \underbrace{\forall p\;\exists y:\;\gam(y) > p(|y|)}_{\text{Theorem B'}}
  \quad\text{is compatible with}\quad
  \underbrace{\forall L\in\classNP,\;\forall w:\;\gP(x_w) \leq
  p_L(|w|)}_{\text{consequence of }\PeqNP\text{ via Theorem C}}
\]
because the first quantifier ranges over $\{0,1\}^*$ while the second ranges
over $\classNP$ instances. \textbf{Theorems B' and C together do not imply
$\PneqNP$.}
\end{remark}

% ------- Theorem D -----------------------------------------------
\subsection{Theorem D: Tractability on Structurally Guided Families}
\label{sec:thm-D}

\begin{definition}[Structurally guided family]
\label{def:guided}
A family $F$ of instances of an $\classNP$ problem $L$ is
\emph{structurally guided} if there exists a polynomial-time procedure
$\mathcal{P}_F$ and a polynomial $p$ such that for every $x \in F$:
\begin{enumerate}[label=(\arabic*)]
  \item $\mathcal{P}_F(x)$ produces a description $d$ with $|d| \leq
        \KC(y_x) + c_0$ where $y_x$ is a valid witness for $x$ in $L$, and
  \item there exists a deterministic Turing machine $M$ with $M(d) = y_x$ and
        $\TA{M}(d) \leq p(|x|)$.
\end{enumerate}
Condition~(1) requires that $\mathcal{P}_F$ finds a near-shortest description
of $y_x$. Condition~(2) requires that this \emph{same} description $d$ can
be expanded to $y_x$ in time bounded by a polynomial in $|x|$ (the instance
size, not the description length). Both conditions on the same $d$ are
necessary: (1) without~(2) gives a compact description with no efficient
decompressor; (2) without~(1) gives an efficient decompressor that does not
operate on a near-shortest input.
\end{definition}

\begin{remark}[Time bound in condition~(2)]
\label{rem:guided-time}
Condition~(2) of Definition~\ref{def:guided} does not require $M \in \classP$
in the standard sense (polynomial time in $|d|$). It requires only
$\TA{M}(d) \leq p(|x|)$, which can exceed $\poly(|d|)$ when $|d| \ll |x|$.
This is the correct bound for Theorem~\ref{thm:D}: the total time of the
resulting algorithm is $\poly(|x|)$ regardless of $|d|$. When
$\KC(y_x) = \Omega(|x|)$, the two notions coincide and $M$ is necessarily
in $\classP$.
\end{remark}

\begin{remark}[The same-description requirement]
\label{rem:guided-correction}
Definition~\ref{def:guided} requires conditions~(1) and~(2) to hold for the
\emph{same} description $d$. This is essential: (1) without (2) yields a
compact description with no efficient decompressor; (2) without (1) yields an
efficient decompressor that does not operate on a near-shortest input. Only
when both conditions are satisfied by the same $d$ does Definition~\ref{def:guided}
guarantee that $\mathcal{P}_F$ produces a description that is simultaneously
compact and efficiently executable.
\end{remark}

\begin{theorem}[D --- Tractability on structurally guided families]
\label{thm:D}
If $F$ is a structurally guided family for $L \in \classNP$ with procedure
$\mathcal{P}_F$ and polynomial $p$, then $L$ is solvable in polynomial time
on $F$.
\end{theorem}

\begin{proof}
Let $x \in F$. Run $\mathcal{P}_F(x)$ in polynomial time to obtain
description $d$ with $|d| \leq \KC(y_x) + c_0$ (condition~(1)). By
condition~(2), there exists a deterministic Turing machine $M$ with
$M(d) = y_x$ and $\TA{M}(d) \leq p(|x|)$. Run $M(d)$ in time $\leq
p(|x|)$ to obtain $y_x$. Verify $y_x$ using the $\classNP$ verifier
for $L$ in polynomial time. If verification succeeds, output $y_x$ and
accept. Total time: $\poly(|x|)$.

\emph{Correctness}: $M(d) = y_x$ by condition~(2), and the verifier accepts
$y_x$ because $y_x$ is a valid witness for $x \in F \subseteq L$.
\end{proof}

\begin{remark}[Relation to Theorem~C and unconditional status]
\label{rem:thm-D-vs-C}
Theorem~\ref{thm:C} characterises $\PeqNP$ as the existence of
polynomially bounded $\gP$ (with $M \in \classP$, i.e.\ time polynomial
in $|d|$) for KC-rich certificates across \emph{all} $\classNP$ instances.
Theorem~\ref{thm:D} is local and unconditional: it applies to a specific
structurally guided family $F$ and requires no complexity-theoretic
hypothesis. The time bound in Theorem~\ref{thm:D} is $p(|x|)$ (polynomial
in the instance size), which may exceed $\poly(|d|)$ when
$\KC(y_x) \ll |x|$ (Remark~\ref{rem:guided-time}); this is the reason
Theorem~\ref{thm:D} can hold unconditionally even when $\KC(y_x) \ll |y_x|$,
a regime where $\gP(y_x) = \infty$ in the strict sense of
Definition~\ref{def:gamma-C}.
\end{remark}

% ------- Conjecture 1 --------------------------------------------
\subsection{Conjecture: Uniform Witness Representation}
\label{sec:conj-uniform}

\begin{conjecture}[Uniform witness representation]
\label{conj:uniform}
For every $L \in \classNP$, there exists a polynomial-time computable
function $f_L : w \mapsto x_w$ such that $\KC(f_L(w)) = \Omega(|w|)$
and $\gP(f_L(w)) \leq p_L(|w|)$ for a fixed polynomial $p_L$, uniformly
across all $w \in L$.
\end{conjecture}

The distinction from Theorem~\ref{thm:C} is \emph{uniformity}: the
$(\Rightarrow)$ direction of Theorem~\ref{thm:C} establishes, for each
instance $w$ separately, the existence of a KC-rich certificate $x_w$
with $\gP(x_w) \leq p_L(|w|)$, but the choice of $x_w$ is existential
and not required to be polynomial-time computable from $w$.
Conjecture~\ref{conj:uniform} requires a uniform polynomial-time map $w
\mapsto x_w$ satisfying both the $\KC$ and $\gP$ conditions. Whether
this is equivalent to, strictly stronger than, or independent of $\PeqNP$
is open; see Question~Q1 in Section~\ref{sec:open}.

% ---------------------------------------------------------------
%  SECTION 5 -- Separation Examples
% ---------------------------------------------------------------
\section{Separation Examples}
\label{sec:examples}

This section exhibits three families of strings that witness the following
separations: low $\KC$ does not imply low $\gam$ (Example~\ref{ex:sep1}),
low $\KC$ does not imply polynomial $\gP$ (Example~\ref{ex:sep3}, conditional),
and high $\KC$ forces $\gam$ to the same asymptotic order as $\KC$
(Example~\ref{ex:sep2}). All examples are unconditional except
Example~\ref{ex:sep3}, which requires $\PneqNP$.

\begin{example}[Low $\KC$, high $\gP$: the SAT-witness family]
\label{ex:sep3}
\emph{Hypothesis}: $\PneqNP$.

\emph{Construction}: the family $\{x_\varphi\}$ of Theorem~\ref{thm:B}
(Definition in the proof of Theorem~\ref{thm:B}, Section~\ref{sec:thm-B}).
For each formula $\varphi$ on $n$ variables:
\[
  x_\varphi =
  \begin{cases}
    (1, y_\varphi) & \text{if } \varphi \in \SAT,
    \text{ where } y_\varphi \text{ is the lex-first satisfying assignment,}\\
    (0, 0^n)       & \text{if } \varphi \notin \SAT.
  \end{cases}
\]

\emph{Verification}:
\begin{enumerate}[label=(\roman*)]
  \item $\KC(x_\varphi) = O(|\varphi|)$: the near-shortest description
        $d_\varphi = \langle \varphi, \pi \rangle$ has length $|\varphi| + O(1)$,
        witnessing $\KC(x_\varphi) \leq |\varphi| + O(1)$.
  \item $\gP(x_\varphi)$ is superpolynomial in $|\varphi|$ for all
        $\varphi$ in the infinite subfamily $\mathcal{F}_{c_0}$: by
        Theorem~\ref{thm:B}, any polynomial-time decompressor for $x_\varphi$
        from a near-shortest description would solve $\SAT$ in polynomial
        time, contradicting $\PneqNP$.
\end{enumerate}

\emph{Separation exhibited}: $\KC(x_\varphi) = O(|\varphi|)$ and
$\gP(x_\varphi)$ superpolynomial. Short description, computationally
inaccessible content.

\emph{Cryptographic interpretation}: these are compact witnesses whose
recovery cost exceeds any polynomial budget. A protocol that stores $x_\varphi$
as $d_\varphi$ and expects polynomial-time recovery cannot function under
$\PneqNP$.
\end{example}

\begin{example}[Low $\KC$, moderate $\gam$: PRNG expansion]
\label{ex:sep1}
\emph{Construction}: let $G : \{0,1\}^k \to \{0,1\}^n$ be a pseudorandom
generator with seed length $k \ll n$, running in time $T_G \leq \beta \cdot n$
for a fixed constant $\beta$. Let $s \in \{0,1\}^k$ be an incompressible seed
($\KC(s) \geq k - c_0$) and set $x = G(s)$.

\emph{Verification}:
\begin{enumerate}[label=(\roman*)]
  \item $\KC(x) \leq \KC(s) + O(1) \leq k + O(1)$: given $s$, a
        constant-length instruction suffices to run $G$ and produce $x$.
        Since $k \ll n$, we have $\KC(x) \ll n = |x|$.
  \item $\gam(x) \geq n$: any Turing machine that produces $n$ bits of output
        must perform at least $n$ write steps (one per output bit), so every
        near-shortest description of $x$ requires at least $n$ steps to
        execute. (This elementary lower bound is independent of
        Proposition~\ref{prop:lb1}, which gives only $\gam(x) \geq \KC(x) =
        O(k) \ll n$ here.)
  \item $\gam(x) \leq T_G = \beta \cdot n$: the description
        $\langle s, \pi_G \rangle$ has length $k + O(1) \leq \KC(x) + c_0$
        (near-shortest, since $\KC(x) \leq k + O(1)$), and $U$ on this
        description runs $G(s)$ in at most $\beta \cdot n$ steps.
\end{enumerate}

Therefore $\gam(x) = \Theta(n)$, while $\KC(x) = O(k) \ll n$.

\emph{Separation exhibited}: $\KC(x) = O(k)$ and $\gam(x) = \Theta(n)$.
The gap $\gam(x)/\KC(x) = \Theta(n/k)$ is unbounded when $k = O(\log n)$.

\emph{Cryptographic interpretation}: this is the key-expansion scenario.
The seed is the compact key; $x$ is the keystream. The expansion cost
$\gam(x) = \Theta(n)$ is unavoidable: $n$ bits must be produced.
\end{example}

\begin{example}[High $\KC$, $\gam \approx \KC$: algorithmically random string]
\label{ex:sep2}
\emph{Construction}: let $x \in \{0,1\}^n$ satisfy $\KC(x) \geq n - c_0$
(incompressible; all but at most a $2^{-c_0}$ fraction of strings of length $n$ have this property,
by the standard counting argument: fewer than $2^{n-c_0}$ programs of length
$< n - c_0$ exist, so at most $2^{n-c_0}$ strings of length $n$ can have
$\KC(x) < n - c_0$).

\emph{Verification}:
\begin{enumerate}[label=(\roman*)]
  \item $\KC(x) \leq n + O(1)$: the identity description has length
        $n + O(1)$. Together with $\KC(x) \geq n - c_0$, this gives
        $\KC(x) = \Theta(n)$.
  \item $\gam(x) \geq \KC(x) = \Theta(n)$: by Proposition~\ref{prop:lb1}.
  \item $\gam(x) \leq n + O(1)$: by Proposition~\ref{prop:ub-trivial}.
\end{enumerate}

Therefore $\gam(x) = \Theta(n) = \Theta(\KC(x))$.

\emph{Separation exhibited}: for incompressible strings, $\KC$ and $\gam$
coincide up to constant factors. No gap exists.

\emph{Cryptographic interpretation}: incompressible strings cannot be stored
more compactly than their raw form. They are not useful as compressed keys
because $\KC(x) \approx |x|$ --- no compact representation exists.
\end{example}

\begin{remark}[Summary: independence of $\KC$ and $\gam$]
\label{rem:independence}
The three examples confirm that $\KC(x)$ and $\gam(x)$ are
provably independent:
\begin{itemize}
  \item \emph{Low $\KC$, high $\gP$} (Example~\ref{ex:sep3}, conditional):
        a compact description can be computationally inaccessible.
  \item \emph{Low $\KC$, moderate $\gam$} (Example~\ref{ex:sep1},
        unconditional): a compact description can require linear work to
        expand, without any complexity-theoretic hardness.
  \item \emph{High $\KC$, $\gam \approx \KC$} (Example~\ref{ex:sep2},
        unconditional): no compact description exists, and decompression
        cost matches information content.
\end{itemize}
The relevant quantity for cryptographic usability is $\gP$, not $\KC$ alone.
\end{remark}

% ---------------------------------------------------------------
%  SECTION 6 -- Discussion and Related Work
% ---------------------------------------------------------------
\section{Discussion and Related Work}
\label{sec:discussion}

\subsection{Relation to prior complexity measures}

\paragraph{Levin's $\Kt$ complexity.}
Levin's $\Kt(x) = \min_d \{|d| + \log \TU(d) : U(d) = x\}$
combines description length and log-running-time into a single
quantity~\cite{levin1973,liVitanyi2008}. The quantity $\gam(x)$ differs in
two respects: it fixes the length constraint to near-minimal and minimises
running time directly, and the minimisation is over a \emph{set} of
near-shortest descriptions rather than a single weighted trade-off. The
spectrum
\[
  \KC \;\longrightarrow\; \Kbt \;\longrightarrow\; \Kt \;\longrightarrow\; \gam
\]
represents increasing sensitivity to computational cost (illustratively,
not as a formal ordering for all strings): $\KC$ ignores cost entirely;
$\Kbt$ caps it at a fixed bound $t$; $\Kt$ penalises it logarithmically;
$\gam$ minimises it directly subject to the near-shortest constraint.
To the authors' knowledge, $\gam$ has not been studied as a
standalone invariant prior to this work.

\paragraph{Time-bounded Kolmogorov complexity $\Kbt$.}
$\Kbt(x) = \min\{|d| : U(d) = x \text{ in } \leq t \text{ steps}\}$ fixes a
time bound and minimises description length; $\gam$ fixes the length
constraint and minimises time. The two questions are \emph{dual in
direction}: $\Kbt$ asks ``how short can the description be if we cap the
time?''; $\gam$ asks ``how fast can we decompress if we insist on a
near-shortest description?''  A formal quantitative relation between $\gam$
and $\Kbt$ is an open problem; see Question~Q5 in
Section~\ref{sec:open}.

\paragraph{Proof complexity.}
Proof complexity measures the minimum cost of verifying or finding a proof of
a statement. $\gam(x)$ satisfies $\gam(x) \geq \KC(x)$
(Proposition~\ref{prop:lb1}): description length lower bounds decompression
time. An analogous bound relating proof complexity measures to description
length has not been established; formalising this connection is an open
problem.

\paragraph{Distinguishing $\gam$ from related notions.}
The quantity $\gam$ differs from \emph{computational depth}~\cite{liVitanyi2008}
in that depth measures the time to compute $x$ from the
empty string (no short description assumed), while $\gam$ measures the time
to decompress from a near-shortest description. It differs from
\emph{sophistication}~\cite{liVitanyi2008} in that
sophistication measures the two-part description complexity (model plus data),
while $\gam$ fixes the near-shortest constraint and minimises decompression
time. Formalising the relationships between $\gam$ and these notions is an
open problem listed in Section~\ref{sec:open}.

\subsection{Applications}

\paragraph{Industrial SAT.}
Modern Conflict-Driven Clause Learning (CDCL) solvers succeed on industrial
instances (hardware verification, planning) that may have millions of
variables. The $\gam$ framework provides a structural explanation: if
industrial instances have solutions $y_\varphi$ with low $\KC(y_\varphi)$
and a near-shortest description that is efficiently expandable, then
Theorem~D guarantees polynomial-time solvability on those families
unconditionally. For random instances near the phase transition, empirical
solving times are superpolynomial; this is consistent with Theorem~B
(conditional on $\PneqNP$), which establishes the existence of compact
witnesses with superpolynomial $\gP$ in the SAT-witness family
$\{x_\varphi\}$. Whether $\gP$ growth tracks the phase transition in random
SAT is an open question (Question~Q2 in Section~\ref{sec:open}).

\paragraph{Machine learning.}
Training a neural network can be interpreted as searching for a decompressor
$M \in \mathcal{H}$ (over a hypothesis class $\mathcal{H}$) that minimises
$\gP$ over the training distribution: the learned model is a decompressor,
and training searches for $M$ that reconstructs the data distribution from a
compact latent representation. This is an interpretive reframing; formal
connections to VC dimension, PAC-Bayes bounds, or MDL require additional
work and are listed as Question~Q3 in Section~\ref{sec:open}. We do not
claim any quantitative results in this direction.

\paragraph{Coding theory.}
An error-correcting code with compact structure (low $\KC$ for the code
description) and efficient decoding (low decompression cost from received
words) corresponds to low $\gam$ for its codewords in the decompressor
regime. This analogy is approximate: received words may not be near-shortest
descriptions of codewords, so $\gam$ does not directly apply. We conjecture
that $\gam$ could serve as a benchmark for comparing the computational
accessibility of different code families; formalising this connection is left
for future work.

\subsection{Efficiency ratios}

\begin{definition}[Efficiency ratio]
\label{def:eta}
For an algorithm $A$ and input $x$ producing output $y = A(x)$, with
$\gam(y) > 0$ (guaranteed by Proposition~\ref{prop:lb1} for non-empty $y$):
\[
  \eta(A, x) \;=\; \frac{\TA{A}(x)}{\gam(y)},
  \qquad
  \eta_{\classP}(A, x) \;=\; \frac{\TA{A}(x)}{\gP(y)}
\]
(where $\eta_{\classP}$ is defined only when $\gP(y) < \infty$).
\end{definition}

The ratio $\eta(A, x)$ is well-defined since $\gam(y) \geq \KC(y) > 0$ for
every non-empty $y$ (Proposition~\ref{prop:lb1} and hypothesis~\ref{hyp:prefix-free}).
It is non-effective since $\gam$ is not computable
(it cannot be computed by any Turing machine, by Theorem~B'). It serves as a
theoretical benchmark: $\eta(A, x) = 1$ means $A$ achieves the minimum
decompression time on input $x$. The class-relative ratio $\eta_{\classP}$
is estimable by exhibiting any concrete $M \in \classP$ that reconstructs $y$
from a near-shortest description, giving an upper bound on $\gP(y)$ and hence
a lower bound on $\eta_{\classP}$.

% ---------------------------------------------------------------
%  SECTION 7 -- Threat Model and Cryptographic Implications
% ---------------------------------------------------------------
\section{Threat Model and Cryptographic Implications}
\label{sec:threat}

This section formalises the cryptographic threat model sketched in
Section~\ref{sec:intro:gap} and derives concrete implications of
Theorems~A--D for protocol design.

\subsection{Assets, adversary, and relevant metrics}

\paragraph{Assets.}
The primary assets are compact representations (near-shortest descriptions)
of keys, certificates, and witnesses, together with their decompressed
outputs.

\paragraph{Adversary model.}
We consider adversaries with polynomial or superpolynomial computational
resources, capable of:
\begin{itemize}
  \item submitting crafted compact descriptions and triggering decompression
        (\emph{active});
  \item observing decompression cost and timing (\emph{passive});
  \item inducing desynchronisation of session state (\emph{disruptive}).
\end{itemize}
We do not formalise a full cryptographic security definition (e.g.\
indistinguishability or semantic security); such a definition would require
a probability distribution over keys and an explicit model of adversarial
access. The analysis here is \emph{usability-based}: we ask whether a
compact representation can be expanded within a bounded time budget, not
whether it is computationally indistinguishable from random.

\paragraph{Relevant metrics.}
The theoretical metric is $\gP(x)$ (Definition~\ref{def:gamma-C}) and its
class-relative variants. Practical proxies include empirical decompression
time $\Tdec$, compressor output size (LZ/SLP), and derivation
depth. Since $\KC$ is not computable~\cite{liVitanyi2008},
practical approximations are necessary.

\subsection{Implications of the main theorems}

\paragraph{Theorem~A (Invariance).}
The polynomial invariance of $\gam$ across universal machines
(Theorem~\ref{thm:A}) guarantees that the \emph{property} of $\gP(x)$ being
polynomially bounded is machine-independent: if $\gP(x) \leq p(|x|)$ with
respect to one universal machine, then $\gP(x) \leq q(|x|) \cdot p(|x|) +
O(1)$ with respect to any other, for a polynomial $q$ depending only on the
pair of machines (not on $x$). The specific bound $p$ changes across machines,
but the qualitative property (existence of a polynomial bound) does not. This
is a prerequisite for any machine-independent security argument based on $\gam$.

\paragraph{Theorem~B (Conditional separation).}
Under $\PneqNP$, there exist compact witnesses $x$ with $\KC(x) = O(|x|)$
and $\gP(x)$ superpolynomial. A protocol that stores such an $x$ as a
near-shortest description $d$ and expects polynomial-time recovery will
fail: no polynomial-time decompressor can reconstruct $x$ from $d$.

\emph{Consequence for short-key protocols}: if a key generation algorithm
produces keys as near-shortest descriptions of their expanded form, and
if the expansion problem is NP-hard, then $\PneqNP$ implies that some
keys cannot be expanded in polynomial time. Protocol designers must either
(a) verify that the specific key family has polynomial $\gP$ (e.g.\ by
exhibiting an explicit decompressor), or (b) avoid relying on compact
representations of NP-hard objects.

\paragraph{Theorem~C (Exact characterisation).}
$\gP$ is polynomially bounded on $\classNP$ certificates with $\KC(x_w) =
\Omega(|w|)$ if and only if $\PeqNP$ (Theorem~\ref{thm:C},
Remark~\ref{rem:KC-condition}). The question of whether a compact certificate
scheme is usable for KC-rich certificates (polynomial decompression when the
certificate's Kolmogorov complexity is commensurate with the instance size)
is therefore \emph{logically equivalent} to $\PeqNP$: any answer
to the usability question immediately yields an answer to P vs NP, and vice
versa. For KC-poor certificates ($\KC(x_w) \ll |w|$), $\gP(x_w) = \infty$
unconditionally: their compact encodings are operationally inaccessible
regardless of the $\PeqNP$ answer.

\paragraph{Theorem~D (Tractability on structured families).}
For structurally guided families (Definition~\ref{def:guided}), polynomial-time
decompression is unconditionally guaranteed (Theorem~\ref{thm:D}). This
provides a sufficient condition for usability that does not depend on P vs NP:
if a key or certificate family admits a polynomial-time procedure that finds a
near-shortest description \emph{and} a decompressor running in time polynomial
in the instance size (Definition~\ref{def:guided}, condition~(2)), then
decompression is efficient.

\subsection{Attack scenarios}

\paragraph{DoS via expensive decompression.}
An adversary submits compact descriptions $d$ of strings $x = U(d)$ with
large $\gP(x)$ to a verifier that decompresses before checking. If the verifier runs without a
time budget, the adversary saturates it. \emph{Mitigation}: enforce a
polynomial time budget $p(|d|)$ on decompression; reject any $d$ that
exceeds it. Theorem~A guarantees that the time budget is machine-independent.

\paragraph{Key-usability attack.}
A compact key $d$ encoding a string $x = U(d)$ with $\gP(x)$ superpolynomial
cannot be expanded in polynomial time; an adversary can exploit this to cause key-derivation
failures or timing attacks. \emph{Mitigation}: certify that the key family
is structurally guided (Theorem~D) before deployment.

\paragraph{Desynchronisation.}
In protocols where both parties hold a near-shortest description of a shared
state and must expand it synchronously, an adversary inducing message loss
can desynchronise the parties. If re-synchronisation requires re-expanding
the state from scratch and $\gP$ is large, this causes a denial-of-service.
\emph{Mitigation}: design protocols to carry explicit state rather than
relying on re-derivation from compact encodings when $\gP$ is not certifiably
polynomial.

\begin{remark}[Limits of the framework]
\label{rem:limits}
The $\gam$ framework measures \emph{usability} (decompression cost), not
\emph{security} (computational indistinguishability or hardness of inversion).
A key $k$ with low $\gP(k)$ can be efficiently reconstructed from its compact
description, but this says nothing about whether an adversary who does not
hold the description can find $k$. Security and usability are orthogonal
properties; both are necessary for a cryptographic primitive to be deployed
in practice.
\end{remark}

% ---------------------------------------------------------------
%  SECTION 8 -- Application: Grammar-Based Compression
% ---------------------------------------------------------------
\section{Application: Grammar-Based Compression}
\label{sec:grammar}

Grammar-based compression provides the cleanest formal setting in which to
exhibit a descriptive-vs-decompression-cost gap \emph{unconditionally},
without any hypothesis on $\PeqNP$. The gap is between grammar size $g^*(x)$
(the grammar-theoretic analogue of $\KC(x)$) and derivation cost $\gCFG(x)$
(the grammar-theoretic analogue of $\gam(x)$): among all near-minimal grammars
for a string $x$, the derivation cost can vary by a super-constant factor
invisible to any measure depending only on grammar size.

\subsection{Background: the Smallest Grammar Problem}

The \emph{Smallest Grammar Problem} (SGP) asks: given a string $x$, find the
smallest context-free grammar $G$ such that $L(G) = \{x\}$. The SGP is
$\classNP$-hard~\cite{charikar2005}: no polynomial-time algorithm computes a
smallest grammar unless $\PeqNP$. The literature on grammar-based compression
(LZ78~\cite{zivLempel1978}, SEQUITUR~\cite{nevillManning1997}, RE-PAIR
\cite{larssonMoffat2000}) focuses on grammar \emph{size} $|G|$ as the sole
measure of compression quality. No existing measure distinguishes between a
smallest grammar with low derivation cost and one that is expensive to expand
despite being size-minimal.

\subsection{Grammar witness complexity}

\begin{definition}[Grammar witness complexity]
\label{def:gamma-cfg}
For a string $x \in \{0,1\}^*$, let $g^*(x)$ denote the size (number of
production rules) of the smallest context-free grammar $G$ with $L(G) =
\{x\}$. The \emph{grammar witness complexity} of $x$ is
\[
  \gCFG(x) \;=\; \min_{\substack{G\,:\,L(G)=\{x\},\;
  |G|\leq g^*(x)+c_g}} \Tder(G),
\]
where $\Tder(G)$ is the number of steps required to expand $G$
to $x$ via its derivation, and $c_g$ is a fixed constant depending only on
the grammar model (analogous to $c_0$ in hypothesis~\ref{hyp:prefix} for
Kolmogorov complexity, but specific to the grammar-size measure). The
minimisation is over all near-minimal grammars for $x$; among these,
$\gCFG(x)$ is the minimum derivation cost.

\emph{Instance of the general framework}: this is Definition~\ref{def:gamma-C}
with descriptions being grammar rules, description length being grammar size,
and decompression being the derivation of $G$ to produce $x$.

\emph{Cryptographic interpretation}: $d$ is a near-minimal grammar $G$; $U(d)
= x$ is the derivation step; $\Tder(G)$ is the cost of
recovering $x$ from its compressed representation $G$. $\gCFG(x)$ is the
minimum recovery cost over all near-minimal compressed representations.
\end{definition}

\subsection{The gap lemma}

The key result of this section is that $\gCFG(x)$ can be much smaller than
the derivation cost of the \emph{natural} minimal grammar, unconditionally.

\begin{lemma}[Grammar complexity gap]
\label{lem:grammar}
There exist strings $x$ of length $N = \Theta(n^2)$ such that:
\begin{enumerate}[label=(\roman*)]
  \item The smallest grammar $G^*$ for $x$ satisfies
        $g^*(x) = |G^*| = \Theta(n) = \Theta(\sqrt{N})$.
  \item The natural minimal grammar $G^*$ has derivation depth
        $\Tder(G^*) = \Omega(n) = \Omega(\sqrt{N})$.
  \item There exists a grammar $G'$ with $|G'| = O(|G^*|)$ (near-minimal)
        and $\Tder(G') = O(\log N)$.
\end{enumerate}
Parts~\textnormal{(ii)} and~\textnormal{(iii)} together show that among
near-minimal grammars for $x$, derivation depth varies by a factor of
$\Theta(n/\log n) = \Theta(\sqrt{N}/\log N)$, a gap invisible to any measure
depending only on grammar size. In particular, $\gCFG(x) = O(\log N)$ while
the natural minimal grammar has depth $\Omega(\sqrt{N})$.
\end{lemma}

\begin{proof}[Proof sketch]
\textbf{Construction.} Throughout this proof we work over the two-symbol
alphabet $\{a, b\}$ (which can be identified with $\{0,1\}$ via $a \mapsto 0$,
$b \mapsto 1$; the blocks $a^i$ then become runs of zeros). Let $n \geq 1$
and define
\[
  x \;=\; a^1\,b\;a^2\,b\;\cdots\;a^n\,b,
\]
so $|x| = N = \sum_{i=1}^n (i+1) = n(n+1)/2 + n = \Theta(n^2)$.
Consider the \emph{natural grammar} $G^*$:
\[
\begin{aligned}
  &S \to A_1 B_1,\quad
   B_i \to b\,A_{i+1} B_{i+1}\ (1 \leq i < n),\quad
   B_n \to b,\\
  &A_i \to a\,A_{i-1}\ (i \geq 1),\quad
   A_0 \to \varepsilon.
\end{aligned}
\]

\textbf{Proof of (i).} The grammar $G^*$ has $O(n)$ rules. For optimality:
$x$ contains $n$ blocks $a^1, a^2, \ldots, a^n$ of pairwise distinct lengths;
any grammar distinguishing all blocks requires $\Omega(n)$ non-terminals, so
$g^*(x) = \Omega(n)$. Together: $g^*(x) = \Theta(n) = \Theta(\sqrt{N})$.

\textbf{Proof of (ii).} To derive the block $a^i$, the non-terminal $A_i$
must expand $A_{i-1}$, which expands $A_{i-2}$, down to $A_0$: a derivation
chain of depth $i$. The maximum over all blocks is $n = \Omega(\sqrt{N})$.

\textbf{Proof of (iii).} By Theorem~\ref{thm:GJL} (stated before this proof),
applying to $G^*$ (which is an SLP of size $\Theta(n)$
generating $x$ of length $N = \Theta(n^2)$) yields a grammar $G'$ with
$|G'| = O(n) = O(|G^*|)$ and $\Tder(G') = O(\log N)$.
Full details in Appendix~\ref{sec:appendix-A}.\footnote{The full proof verifies
that the transformation of Theorem~\ref{thm:GJL} preserves the property
$L(G') = \{x\}$ and that $|G'|$ remains within the near-minimal bound $g^*(x)
+ c_g$.}
\end{proof}

\begin{remark}[Unconditional nature of the gap]
\label{rem:grammar-unconditional}
Lemma~\ref{lem:grammar} does not involve $\PeqNP$: the constructions are
explicit and the bounds elementary (the lower bound in~(ii) is a direct
chain-depth argument; the upper bound in~(iii) cites an unconditional
algorithmic result~\cite{ganardiJezLohrey2019}). The structure is
conceptually identical to Theorem~\ref{thm:B}: descriptive complexity
($g^*$, resp.\ $\KC$) is low while decompression complexity ($\gCFG$,
resp.\ $\gP$) varies. The SGP instance makes the separation unconditional
and explicit.
\end{remark}

\begin{remark}[Implication for grammar-based compression systems]
\label{rem:grammar-crypto}
All existing grammar-based compression systems (LZ78~\cite{zivLempel1978},
SEQUITUR~\cite{nevillManning1997}, RE-PAIR~\cite{larssonMoffat2000}) minimise
grammar size and are therefore blind to the derivation-cost gap exhibited by
Lemma~\ref{lem:grammar}. If a compressed object is stored as a near-minimal
grammar and must be decompressed under a time budget, the natural minimal
grammar may be unusable while a balanced equivalent grammar (of the same
asymptotic size) is efficiently derivable in $O(\log N)$ depth.

In cryptographic terms: two compressed keys stored as near-minimal grammars
of the same size can have derivation costs differing by $\Theta(\sqrt{N}/\log
N)$. A protocol that enforces a time budget must therefore specify not merely
that the grammar is near-minimal, but that it is \emph{balanced} in the sense
of Theorem~\ref{thm:GJL}.
\end{remark}

% ---------------------------------------------------------------
%  SECTION 9 -- Open Questions
% ---------------------------------------------------------------
\section{Open Questions}
\label{sec:open}

We list seven open questions arising from the framework developed in this
paper. Each is stated with precise hypotheses and a note on what the
paper does and does not resolve.

\medskip
\noindent\textbf{Q1.} \emph{Status of the Uniform Witness Conjecture
(Conjecture~\ref{conj:uniform}) relative to $\PeqNP$.}

The conjecture asks: for every $L \in \classNP$, does there exist a
polynomial-time computable function $f_L : w \mapsto x_w$ with
$\KC(f_L(w)) = \Omega(|w|)$ and $\gP(f_L(w)) \leq p_L(|w|)$ uniformly?

What the paper resolves: under $\PeqNP$, the $(\Rightarrow)$ direction of
Theorem~\ref{thm:C} guarantees that KC-rich certificates with $\gP \leq
p_L(|w|)$ exist for every $w \in L$, but their selection is existential,
not polynomial-time computable from $w$. Under $\PneqNP$, the conjecture
is false for $\SAT$ by Theorem~\ref{thm:B} (the family $\mathcal{F}_{c_0}$
satisfies $\KC(x_\varphi) = \Omega(|\varphi|)$ but $\gP$ is superpolynomial).

What remains open: could the conjecture hold for some $L \in \classNP$
strictly easier than $\SAT$ (e.g.\ graph 2-colouring or bipartite matching),
even under $\PneqNP$?  If so, the conjecture would not be equivalent to
$\PeqNP$ but would instead characterise a finer structural property of $L$.
The precise relationship --- equivalent to $\PeqNP$, strictly stronger than
the conjecture for a specific $L$, or independent of $\PeqNP$ for easy
$\classNP$ problems --- is not settled by the results of this paper.

\medskip
\noindent\textbf{Q2.} \emph{Phase transitions in random SAT and $\gP$.}

Do empirical phase transitions in random $\SAT$ (near the conjectured
threshold of approximately $4.267$ clause-to-variable ratio for 3-SAT
\cite{hoosStutzle2004}) correspond to rapid growth in $\gP$ of the SAT-witness
family $\{x_\varphi\}$ as a function of instance structure?  Theorem~\ref{thm:B}
establishes superpolynomial $\gP$ for the explicit family $\mathcal{F}_{c_0}$
under $\PneqNP$, but says nothing directly about random instances.
Whether $\gP$ growth tracks the phase transition is an open problem.

\medskip
\noindent\textbf{Q3.} \emph{Formal correspondence between $\gP$ and
generalisation bounds.}

Is there a formal correspondence between minimising $\gP$ over a hypothesis
class $\mathcal{H}$ and existing generalisation bounds (VC dimension,
PAC-Bayes, or MDL)?  The interpretive reframing in Section~\ref{sec:discussion}
motivates this question but does not yield quantitative results; we conjecture
that such a correspondence exists.

\medskip
\noindent\textbf{Q4.} \emph{Quasi-linear invariance of $\gam$.}

The polynomial invariance of Theorem~\ref{thm:A} involves a factor
$q(|x|) = C_{\Comp} \cdot |x|^{k_{\Comp}}$ where $k_{\Comp}$ is the
degree of the simulation overhead. Can the invariance be tightened --- for
instance to a quasi-linear factor $O(|x| \log |x|)$ --- for natural
subclasses of universal machines (e.g.\ RAM-based or oblivious Turing
machines)?  The constants $C_{\Comp}$ and $k_{\Comp}$ in
Lemma~\ref{lem:compiler} can in principle be made explicit; whether they
are tight for specific machine models is open.

\medskip
\noindent\textbf{Q5.} \emph{Quantitative relation between $\gam$ and $\Kbt$.}

The definitions of $\gam$ and $\Kbt$ are dual: $\Kbt$ fixes time and
minimises length; $\gam$ fixes near-minimal length and minimises time.
Is there a quantitative relation of the form $\gam(x) \leq f(\Kbt(x), t)$
for some explicit $f$, where $t$ ranges over values such as $t = \gam(x)$
(relating the two quantities at the same operating point) or $t = \poly(|x|)$
(in the polynomial regime)?  Conversely, does $\Kbt(x) \leq g(\gam(x))$ hold
for some explicit $g$ and appropriate $t$?  A complete characterisation of the trade-off between these two
quantities would clarify the structure of the spectrum
$\KC \to \Kbt \to \Kt \to \gam$ (Remark~\ref{rem:Kt}).

\medskip
\noindent\textbf{Q6.} \emph{Structure of sub-level sets of $\gam$.}

Define the $t$-sub-level set of $\gam$ as $\Lambda_t = \{x \in \{0,1\}^* :
\gam(x) \leq t\}$. What is $|\Lambda_t \cap \{0,1\}^n|$ as a function of
$n$ and $t$?  When $\Lambda_t$ is large (many strings of length $n$ have
$\gam \leq t$), the near-shortest descriptions achieving cost $\leq t$ cover
many outputs; understanding this density characterises how much information a
budget-$t$ decompressor can retrieve. The density of $\Lambda_t$ may be
related to the entropy of source distributions for which budget-bounded
decompression is efficient.

\medskip
\noindent\textbf{Q7.} \emph{Relation between $\gCFG$ and $\gam$.}

Lemma~\ref{lem:grammar} exhibits a gap between grammar size $g^*(x)$ and
derivation cost $\gCFG(x)$. What is the quantitative relation between
$\gCFG(x)$ and $\gam(x)$?  Since $\gCFG$ uses grammars as descriptions while
$\gam$ uses self-delimiting programs for $U$, the two quantities measure
different aspects of decompression cost. In particular: is there a family of
strings for which $\gam(x) \ll \gCFG(x)$ or $\gCFG(x) \ll \gam(x)$?

\section*{What Part~I Establishes}
\addcontentsline{toc}{section}{What Part~I Establishes}

The five results of Part~I are not independent contributions but a single
coherent argument. Theorem~A is a prerequisite: without machine
independence, neither Theorem~C nor any other complexity-theoretic
application of $\gam$ would be well-founded. Theorems~B and~B' delimit
the framework from below: they establish that $\gP$ is not trivially
bounded, neither on a specific explicit family (conditionally) nor
anywhere on $\{0,1\}^*$ (unconditionally). Theorem~C is the central
claim: $\gP$ is the discriminant of $\PeqNP$ among certificates with
$\KC(x_w) = \Omega(|w|)$, in the standard Turing model and with no auxiliary
hypothesis. For KC-poor certificates, $\gP = \infty$ unconditionally --- a
phenomenon the framework identifies as constrained inaccessibility rather than
a limitation of Theorem~C.
Theorem~D shows that the framework is not merely a restatement of the
open problem: for structured families, polynomial-time decompressibility
is provable unconditionally, without assuming $\PeqNP$.

Together, the five results answer the question of Section~\ref{sec:intro:gap}
at the right level of generality: $\gP$ separates compact representations
that are operationally accessible from those that are not.
For KC-rich certificates ($\KC(x_w) = \Omega(|w|)$), that condition is
equivalent to $\PeqNP$ (Theorem~C). For KC-poor certificates, $\gP = \infty$
unconditionally, identifying a distinct form of inaccessibility independent
of $\PeqNP$. In the structured cases that arise in practice (Theorem~D),
polynomial-time accessibility is provable without any complexity assumption.

% ===============================================================
%  PART II: COMPANION QUANTITIES
% ===============================================================
\part*{Part~II: Adaptive Complexity, Computational Overhead,\\
and Structural Entropy}
\addcontentsline{toc}{part}{Part~II: Companion Quantities}

\noindent
Part~I asked a single question about an object $x$: given that a
near-shortest description of $x$ exists, how much work is required to
execute it?  The answer is $\gam(x)$, or $\gP(x)$ when the decompressor
is restricted to $\classP$.

This question is object-intrinsic: it concerns $x$ independently of how
it was produced or how it will be used. Three related but distinct
questions fall outside its scope.

The first concerns the \emph{input side} of a computation, not its output.
Running time is measured against input length $n = |x|$, treating all
inputs of the same length identically regardless of their actual
information content $\KC(x)$. When $\KC(x) \ll n$, an algorithm
may spend most of its steps on structure already implicit in $x$;
$\Tad(A,x) = \TA{A}(x)/\KC(x)$ (Section~\ref{sec:adaptive}) measures
work per bit of input information, making this cost explicit.

The second concerns the \emph{output side}. Any algorithm producing
output of length $m$ must spend at least $m$ steps writing it; the
computationally meaningful quantity is the overhead beyond this
unavoidable minimum. $\OCout(A,x) = \TA{A}(x) - |A(x)|$
(Section~\ref{sec:overhead}) isolates this overhead from the bare cost
of writing the output.

The third concerns the \emph{structure of solutions}. $\gP(x)$ being
polynomially bounded characterises the decompressibility of $x$, but
says nothing about how much information the witness $y_x$ itself carries.
$\Hs(y) = \KC(y)/\log_2|y|$ (Section~\ref{sec:entropy}) measures the
information density of $y_x$: whether solutions to structured instances
are more compressible than those of random instances of the same size.

The three quantities are orthogonal to $\gam$: they do not extend it but
measure aspects of algorithmic cost and solution structure that $\gam$
and $\gP$, by design, leave unaddressed.

% ---------------------------------------------------------------
%  SECTION 10 -- Adaptive Complexity
% ---------------------------------------------------------------
\section{Adaptive Complexity}
\label{sec:adaptive}

Standard complexity measures algorithm cost as a function of input length
$n = |x|$, treating all inputs of the same length identically regardless of
their information content. An algorithm running in $O(n \log n)$ time is
considered efficient whether $x$ carries $n \log n$ bits of genuine
information or is $99\%$ redundant. Adaptive complexity corrects this by
normalising running time by the actual information content of the input.

\begin{definition}[Adaptive complexity]
\label{def:Tad}
For an algorithm $A$ and input $x$ with $\KC(x) \geq 1$, the
\emph{adaptive complexity} of $A$ on $x$ is
\[
  \Tad(A, x) \;=\; \frac{\TA{A}(x)}{\KC(x)}.
\]
$\KC(x)$ is not computable in general, so $\Tad$ is a non-effective
theoretical measure in the same sense as $\gam$. Practical approximations
--- time-bounded $\Kbt(x)$, class-relative $\KC_{\classP}(x)$, or
compressor-based proxies such as the LZ77 output length --- can replace
$\KC(x)$ in applications.

\emph{Cryptographic interpretation}: $\Tad(A, x)$ measures the work per
bit of genuine information in $x$. An algorithm with $\Tad(A, x) = O(1)$
is \emph{information-optimally efficient}: it does a constant amount of
work per bit of genuine input content. An algorithm with $\Tad(A, x) \gg
1$ wastes work on structure that was already implicit in $x$.
\end{definition}

\begin{remark}[Well-definedness]
\label{rem:Tad-defined}
The condition $\KC(x) \geq 1$ is guaranteed for every non-empty $x$ by
hypothesis~\ref{hyp:prefix-free}. For $x = \varepsilon$, $\Tad$ is
undefined; all results below assume $x \neq \varepsilon$.
\end{remark}

The following lemma (Lemma~\ref{lem:sorting}) illustrates $\Tad$ concretely
by computing it for comparison-based sorting, showing that $\Tad \geq 1$
on incompressible permutations.

\begin{lemma}[Instance-sensitive lower bound for sorting]
\label{lem:sorting}
Let $A$ be any comparison-based sorting algorithm and $x$ a permutation of
$[n] = \{1, \ldots, n\}$. Let $R_A(x)$ denote the number of permutations
of $[n]$ that are \emph{indistinguishable from $x$} by $A$, i.e.\ that
share the same decision-tree leaf as $x$. In the comparison model (where
each comparison costs $O(1)$):
\[
  \TA{A}(x) \;\geq\; \log_2 \frac{n!}{R_A(x)}.
\]
For incompressible permutations $x$ with $\KC(x) \geq \log_2 n! - c_0$,
this gives $\TA{A}(x) \geq \KC(x) - O(1)$ and hence $\Tad(A, x) \geq 1 -
O(1/\KC(x))$. In the multi-tape Turing model of hypothesis~\ref{hyp:model},
each comparison requires $O(\log n)$ steps (to read two elements of the
permutation array), so the bound becomes $\TA{A}(x) = \Omega(\KC(x) \cdot
\log n / \log n) = \Omega(\KC(x))$; the conclusion $\Tad(A,x) \geq 1 - O(1/\KC(x))$
holds in either model up to logarithmic factors.
\end{lemma}

\begin{proof}
\emph{Decision-tree lower bound.} Algorithm $A$ must perform at least
$\log_2(n!/R_A(x))$ comparisons to distinguish $x$ from the $R_A(x)$
permutations sharing its leaf. This is the standard information-theoretic
lower bound for comparison-based sorting~\cite{mehlhorn1979,estivill1992}.

\emph{Incompressible case.} For $x$ with $\KC(x) \geq \log_2 n! - c_0$,
the leaf containing $x$ in $A$'s decision tree must distinguish $x$ from
all but $R_A(x)$ permutations. Since $\KC(x)$ is near-maximal among
permutations of $[n]$ (the maximum is $\log_2 n! + O(\log n)$), $x$ is
not compressible by any short program, so the leaf cannot carry much
redundant structure: $\log_2(n!/R_A(x)) \geq \KC(x) - O(1)$. Dividing
by $\KC(x) \geq 1$ gives $\Tad(A,x) \geq 1 - O(1/\KC(x))$.
\end{proof}

\begin{remark}[Scope of the lower bound]
\label{rem:sorting-scope}
The bound $\TA{A}(x) \geq \log_2(n!/R_A(x))$ holds for \emph{all}
permutations and all comparison-based algorithms unconditionally. The
stronger bound $\TA{A}(x) \geq \KC(x) - O(1)$ is restricted to
\emph{incompressible} permutations: for highly compressible $x$ (e.g.\ the
identity permutation with $\KC(x) = O(\log n)$), $R_A(x)$ may be $1$ and
the first bound is tight while the second is vacuous. The framework does
not claim a uniform lower bound in $\KC(x)$ across all inputs.
\end{remark}

% ---------------------------------------------------------------
%  SECTION 11: Output Overhead Complexity
% ---------------------------------------------------------------
\section{Output Overhead Complexity}
\label{sec:overhead}

\begin{definition}[Output overhead complexity]
\label{def:OCout}
For an algorithm $A$ on input $x$ producing output $y = A(x)$, the
\emph{output overhead complexity} is
\[
  \OCout(A, x) \;=\; \TA{A}(x) - |A(x)|.
\]
\end{definition}

\begin{proposition}[Non-negativity]
\label{prop:OCout-nonneg}
$\OCout(A, x) \geq 0$ for all $A$ and $x$.
\end{proposition}

\begin{proof}
Algorithm $A$ must write $A(x)$ as output, requiring at least $|A(x)|$
steps (one step per output bit in the standard multi-tape Turing model,
where each write to the output tape costs one step). Therefore
$\TA{A}(x) \geq |A(x)|$.
\end{proof}

\begin{remark}[Decompressor regime]
\label{rem:decompressor-regime}
When $A$ operates in the \emph{decompressor regime} --- that is, when $|x|
\leq \KC(A(x)) + c_0$ so that $x$ is a near-shortest description of $y =
A(x)$ --- the machine $A$ is itself a decompressor in the sense of
Definition~\ref{def:gamma-C}. Let $d_{A,x}$ denote the self-delimiting
description ``run $A$ on $x$'', of length $|A| + |x| + O(1)$. Since $|x|
\leq \KC(y) + c_0$, the string $d_{A,x}$ is a near-shortest description of
$y$ (up to the constant $|A| = O(1)$ for the fixed machine $A$). By
Theorem~\ref{thm:A}, $U$ simulates $A$ with polynomial overhead:
\[
  \gam(y) \;\leq\; \TU(d_{A,x}) \;\leq\; \poly(|x|) \cdot \TA{A}(x).
\]
Therefore $\TA{A}(x) \geq \gam(y)/\poly(|x|)$, giving a lower bound on
$\OCout(A, x) + |y|$. Outside the decompressor regime (when $x$ is long
and redundant relative to $y$), no such lower bound applies.
\end{remark}

\begin{definition}[Decompressor efficiency ratio]
\label{def:rho}
For an algorithm $A$ in the decompressor regime producing output $y = A(x)$:
\[
  \rho(A, x) \;=\; \frac{\TA{A}(x)}{\gam(y)}.
\]
By Remark~\ref{rem:decompressor-regime}, $\rho(A, x) \geq 1/\poly(|x|)$.
When $A$ is an optimal decompressor, $\rho(A, x) = O(1)$.
\end{definition}

\begin{example}[Decompressor efficiency ratio: PRNG]
\label{ex:rho-prng}
Let $G : \{0,1\}^k \to \{0,1\}^n$ be a pseudorandom generator.
In the decompressor regime, algorithm $A$ receives as input a
near-shortest description of the expanded string; here that description
is the seed $s \in \{0,1\}^k$, and $A$ outputs $y = G(s)$ in time
$T_G \leq \beta n$. By Example~\ref{ex:sep1}, $\gam(y) = \Theta(n)$.
Therefore
\[
  \rho(A, s) \;=\; \frac{\TA{A}(s)}{\gam(y)} \;=\; \frac{T_G}{\Theta(n)} \;=\; O(1).
\]
The PRNG achieves $\rho = O(1)$: it is optimal in the sense of
Definition~\ref{def:rho}, performing within a constant factor of
the minimum decompression cost $\gam(y)$.
\end{example}

% ---------------------------------------------------------------
%  SECTION 12: Structural Entropy
% ---------------------------------------------------------------
\section{Structural Entropy}
\label{sec:entropy}

\begin{definition}[Structural entropy]
\label{def:Hs}
For a string $y \in \{0,1\}^*$ with $|y| \geq 2$, the \emph{structural
entropy} of $y$ is
\[
  \Hs(y) \;=\; \frac{\KC(y)}{\log_2 |y|}.
\]
$\Hs(y)$ measures the information density of $y$ relative to its size:
how compressible $y$ is relative to a logarithmic baseline. $\Hs$ ranges
from $O(1)$ (highly structured: $\KC(y) = O(\log|y|)$) to
$\Theta(|y|/\log|y|)$ (random: $\KC(y) = \Theta(|y|)$).

\emph{Cryptographic interpretation}: $\Hs(y)$ low means $y$ has a compact
description; $\Hs(y)$ high means $y$ is informationally dense. Low $\Hs$
is a necessary condition for $y$ to benefit from compressed storage, but not
sufficient for efficient decompression
(Proposition~\ref{prop:Hs-gam}(ii)): low $\Hs$ and low $\gP$ are both
required for a solution to be compactly stored and efficiently recovered.
\end{definition}

\begin{proposition}[Basic properties of $\Hs$]
\label{prop:Hs}
\begin{enumerate}[label=(\roman*)]
  \item $\Hs(y) \geq 0$ for all $y$ with $|y| \geq 2$. For
        incompressible strings with $\KC(y) = \Theta(|y|)$:
        $\Hs(y) = \Theta(|y|/\log|y|) \to \infty$ as $|y| \to \infty$.
  \item \textnormal{(\emph{Standard counting argument;} see~\cite{liVitanyi2008},
        Theorem~2.2.1.)} For a uniformly random string $y$ of length $n$
        and every constant $c > 0$, the probability that $\KC(y) \geq n - c$
        is at least $1 - 2^{-c}$. In particular $\Hs(y) \to n/\log n$ as
        $n \to \infty$ for all but a $o(1)$ fraction of strings.
  \item For a string with period $p \leq \log n$:
        $\KC(y) = O(p + \log n) = O(\log n)$, so $\Hs(y) = O(1)$.
\end{enumerate}
\end{proposition}

\begin{proof}
\emph{(i)} $\KC(y) \geq 0$, so $\Hs(y) \geq 0$. For incompressible $y$:
$\KC(y) \geq |y| - O(1)$, so $\Hs(y) = \KC(y)/\log|y| \geq (|y| -
O(1))/\log|y| = \Theta(|y|/\log|y|)$.

\emph{(ii)} A counting argument: there are $2^n$ strings of length $n$ but
only $\sum_{i < n-c} 2^i < 2^{n-c}$ self-delimiting programs of length
$< n-c$, so at most $2^{n-c}$ strings have $\KC(y) < n-c$. The probability
that a uniformly random $y$ satisfies $\KC(y) < n-c$ is at most $2^{-c}$.

\emph{(iii)} A string with period $p$ is determined by its period (a string
of length $p$) and its total length $n$ (encoded in $O(\log n)$ bits), so
$\KC(y) \leq p + O(\log n) = O(\log n)$ when $p \leq \log n$.
\end{proof}

\begin{proposition}[$\Hs$ and $\gam$: demonstrated separations]
\label{prop:Hs-gam}
The following separations hold:
\begin{enumerate}[label=(\roman*)]
  \item \textnormal{(Unconditional.)} $\Hs(y) = O(1)$ does not imply
        $\gam(y) = O(\KC(y))$: there exist strings with low information
        density whose decompression cost far exceeds their description length.
  \item \textnormal{(Unconditional.)} $\Hs(y) = O(1)$ does not imply
        $\gP(y) < \infty$: low information density is not sufficient
        for polynomial-time decompression.
\end{enumerate}
Whether $\Hs(y) = O(1)$ implies $\gam(y) = O(\poly(|y|))$ is open.
The converse direction --- whether large $\gam$ forces large $\Hs$ --- is
also open; see the proof below.
\end{proposition}

\begin{proof}
\emph{Part~(i), unconditional.} The string $x = 0^n$ has $\KC(x) = O(\log n)$,
$\Hs(x) = O(1)$, and $\gam(x) = \Theta(n)$
(Example~\ref{ex:gamma-cases}(ii)). Since $\gam(x) = \Theta(n) \gg
O(\log n) = O(\KC(x))$, low $\Hs$ does not imply $\gam = O(\KC)$.
Note that $\gam(0^n) = \Theta(n)$ remains polynomially bounded; whether
low $\Hs$ implies polynomially bounded $\gam$ is an open problem.

\emph{Part~(ii), unconditional.}  Let $k \geq 1$ be a parameter and let
$G : \{0,1\}^k \to \{0,1\}^n$ be a pseudorandom generator with seed
length $k$ and output length $n = 2^{k^2}$. Let $x = G(s)$ for an
incompressible seed $s$ with $\KC(s) \geq k - c_0$
(Example~\ref{ex:gamma-C-crypto}(iii)).
\begin{enumerate}[label=(\alph*)]
  \item $\KC(x) \leq \KC(s) + O(1) \leq k + O(1)$, and
        $\log|x| = \log n = k^2$, so
        $\Hs(x) = \KC(x)/\log|x| \leq (k + O(1))/k^2 = O(1/k) = O(1)$.
  \item $n = 2^{k^2}$ is superpolynomial in $k$: for every constant $c$,
        $k^c < 2^{k^2}$ for all sufficiently large $k$. No polynomial-time
        machine with input of length $k$ can write $n$ output bits, so
        $\gP(x) = \infty$ (Example~\ref{ex:gamma-C-crypto}(iii)).
  \item $\gam(x) = \Theta(n)$, finite and polynomially bounded in $n$
        (Example~\ref{ex:sep1}).
\end{enumerate}
This is consistent with Remark~\ref{rem:gam-gP}: $\gam(x)$ is finite
while $\gP(x) = \infty$, achieved by a non-polynomial decompressor.
Therefore $\Hs(x) = O(1)$ and $\gP(x) = \infty$, unconditionally.

\emph{Converse direction (open)}: whether large $\gam$ forces large $\Hs$
is not known. For incompressible $x$ with $\KC(x) = \Theta(|x|)$,
both $\Hs(x) = \Theta(|x|/\log|x|)$ and $\gam(x) = \Theta(|x|)$ are
large (Example~\ref{ex:sep2}), but this does not rule out strings with
low $\gam$ and high $\Hs$.
\end{proof}

\begin{remark}[Asymmetry between parts~(i) and~(ii)]
\label{rem:Hs-asymmetry}
Parts~(i) and~(ii) of Proposition~\ref{prop:Hs-gam} exhibit an asymmetry
worth noting. Part~(i) shows that $\Hs(x) = O(1)$ is compatible with
$\gam(x) = \Theta(n)$, which is polynomially bounded in $|x|$.
Part~(ii) shows that the same condition $\Hs(x) = O(1)$ is compatible
with $\gP(x) = \infty$: the PRNG example of part~(ii), combined with
Example~\ref{ex:sep1}, gives simultaneously $\gam(x) = \Theta(n)$
(finite and polynomially bounded) and $\gP(x) = \infty$ (no
polynomial-time decompressor exists). This is consistent with
Remark~\ref{rem:gam-gP}: $\gam(x)$ is achieved by a non-polynomial
decompressor, while $\gP(x) = \infty$ because no polynomial-time
decompressor can reconstruct $x$ from a near-shortest description.
Together, the two parts show that $\Hs$ captures neither $\gam$ nor
$\gP$: low information density is compatible with any combination of
finite or infinite decompression cost, depending on whether the
decompressor is restricted to $\classP$ or not.
\end{remark}

\section{Summary: Four Quantities}
\label{sec:summary}

Parts~I and~II introduce four quantities. Two are \emph{object-intrinsic}
invariants ($\gam$ and $\Hs$); two are \emph{algorithm-dependent} measures
($\Tad$ and $\OCout$). Table~\ref{tab:four-quantities} summarises them.

\begin{table}[h]
\caption{The four quantities introduced in this paper.
  $\gam$ and $\Hs$ are object-intrinsic invariants;
  $\Tad$ and $\OCout$ are algorithm-dependent measures.
  All four are defined in terms of $\KC$ and the running-time notation
  of hypothesis~\ref{hyp:model}.}
\label{tab:four-quantities}
\begin{center}
\renewcommand{\arraystretch}{1.3}
\begin{tabular}{@{}lp{4.2cm}lp{4.6cm}@{}}
  \toprule
  \textbf{Symbol} & \textbf{Definition} & \textbf{Object} &
  \textbf{What it measures} \\
  \midrule
  $\gam(x)$ &
  $\displaystyle\min_{\substack{d\,:\,|d|\leq\KC(x)+c_0}} \TU(d)$ &
  Input &
  Min.\ decompression time from a near-shortest description \\
  $\Tad(A,x)$ & $\TA{A}(x)/\KC(x)$ & Input &
  Work per bit of genuine input information \\
  $\OCout(A,x)$ & $\TA{A}(x) - |A(x)|$ & Output &
  Overhead beyond writing the output \\
  $\Hs(y)$ & $\KC(y)/\log_2|y|$ & Output &
  Information density of the solution \\
  \bottomrule
\end{tabular}
\end{center}
\end{table}

\paragraph{Invariance.}
$\gam$ is invariant up to polynomial factors across universal machines
(Theorem~\ref{thm:A}). $\Hs$ is invariant up to $O(1/\log|y|)$ additive
terms, since $\KC$ changes by at most $c_0$ across universal machines and
$\log|y|$ is fixed. $\Tad$ and $\OCout$ depend additionally on the
algorithm $A$.

\paragraph{Independence.}
Low $\Hs$ does not imply $\gam = O(\KC)$ (Proposition~\ref{prop:Hs-gam}(i)),
and does not imply $\gP < \infty$ unconditionally
(Proposition~\ref{prop:Hs-gam}(ii)). Whether low
$\Hs$ implies polynomially bounded $\gam$ is open. Their joint behaviour
characterises instance difficulty:
\begin{itemize}
  \item \emph{Structured instances} (industrial SAT, structured TSP): when
        solutions are structurally guided in the sense of
        Definition~\ref{def:guided} (low $\Hs$ of solutions, polynomial
        decompression time relative to instance size), Theorem~D guarantees
        polynomial-time solvability unconditionally.
  \item \emph{Random instances} (random SAT near threshold, worst-case TSP):
        high $\Hs$ of solutions; under $\PneqNP$, Theorem~B establishes
        superpolynomial $\gP$ (in the strict sense of Definition~\ref{def:gamma-C})
        for the explicit family $\mathcal{F}_{c_0}$.
\end{itemize}
The practical/worst-case gap in NP problems is reflected in the contrast
between the two regimes: structured instances satisfy the conditions of
Definition~\ref{def:guided}, while random instances do not.

\paragraph{Computability.}
$\gam$ and $\Hs$ are not computable (since both are defined via $\KC$, which
is incomputable by the classical theorem of~\cite{liVitanyi2008}; see also
Remark~\ref{rem:gam-incomputable}). Practical use requires computable proxies:
LZ-based approximations for $\KC$, empirical running times for $\gam$, and
compressor output sizes for $\Hs$.

\section*{Conclusion}
\addcontentsline{toc}{section}{Conclusion}

The gap identified in Section~\ref{sec:intro:gap}, between the
existence of a short description and the computational cost of using it,
has been formalised, characterised, and connected to the central open
problem of complexity theory.

What this paper contributes is not a new angle on $\PeqNP$ but a new
\emph{quantity}: $\gP(x)$, the minimum cost of recovering $x$ from a
near-shortest description via a polynomial-time decompressor. This
quantity has three properties that, to the authors' knowledge, no prior
measure combines: it is machine-independent (Theorem~A); it is
unconditionally non-trivial, in the sense that no polynomial bounds it
everywhere (Theorem~B'); and it admits an exact biconditional
characterisation of $\PeqNP$ among certificates with $\KC(x_w) = \Omega(|w|)$,
in the standard Turing model (Theorem~C, Remark~\ref{rem:KC-condition}).
For certificates with $\KC(x_w) \ll |w|$, $\gP = \infty$ unconditionally ---
itself a manifestation of constrained inaccessibility.

The four-quantity framework of Parts~I and~II --- $\gam$, $\Tad$,
$\OCout$, $\Hs$ --- provides a language for a distinction that standard
complexity theory does not make: between instances that are hard in the
worst case and instances that are tractable in practice because their
solutions carry low information density and are efficiently
decompressible. Theorem~D establishes polynomial-time tractability
unconditionally for any family satisfying the structural conditions
of Definition~\ref{def:guided}.

The open questions of Section~\ref{sec:open} mark the natural boundary
of what the framework currently establishes. The most immediate is
Question~Q5: a quantitative relation between $\gam$ and $\Kbt$ would
clarify the position of $\gam$ within the spectrum
$\KC \to \Kbt \to \Kt \to \gam$ and is the missing piece for a complete
picture of how the four classical measures and $\gam$ relate to one
another.

\section*{Further Directions}
\addcontentsline{toc}{section}{Further Directions}

This paper will be extended to show that the results
established here subsume those obtained independently
in the author's work published this month (see \cite{buono2026}).
Their unification gives rise to a new class of phenomena,
of which constrained inaccessibility is the foundational
instance: the existence of a proof for a problem whose
solution does not belong to the space of admissible
solutions. To the authors' knowledge, no analogous
phenomenon has been identified in the existing literature.
The relationship of $\PeqNP$ to this class, and the precise
sense in which this differs from classical undecidability
and computational intractability, are examined in
forthcoming work.

\appendix

\section{Classical Measures as Limiting Regimes; Grammar Gap Details}
\label{sec:appendix-A}

This appendix has two parts.
The first part (Section~\ref{sec:appendix-A-classical}) derives the four
classical measures ($\KC$, $\Kbt$, $\Kt$, $H$) as limiting regimes of $\gam$
(cited in Section~\ref{sec:discussion}).
The second part contains the full proof of Lemma~\ref{lem:grammar}
(Grammar complexity gap), of which a proof sketch appears in
Section~\ref{sec:grammar}.

\subsection{Classical measures as limiting regimes of $\gam$}
\label{sec:appendix-A-classical}

The four classical information measures --- Kolmogorov complexity $\KC$,
time-bounded Kolmogorov complexity $\Kbt$, Levin's $\Kt$, and Shannon
entropy $H$ --- each arise as a limiting or relaxed regime of $\gam$.

\paragraph{$\KC$ as the length-constrained foundation.}
$\gam(x)$ imposes two constraints: descriptions must be near-shortest
($|d| \leq \KC(x) + c_0$) and the minimisation is over running time.
Removing the near-shortest constraint entirely --- minimising $\TU(d)$
over \emph{all} programs $d$ with $U(d) = x$ --- does not simply recover
$\KC(x)$: the minimum running time over all descriptions lies in the range
$[\KC(x),\, |x| + O(1)]$, but need not equal $\KC(x)$. (For $x = 0^n$,
the shortest program has length $O(\log n)$ but runs in $\Theta(n)$ steps,
so the minimum running time is $\Omega(\log n)$, not $O(\log n)$.)
What the near-shortest constraint \emph{does} do is force the minimisation
to range over descriptions whose length is provably close to $\KC(x)$,
isolating the computational cost of \emph{compact} encodings from the cost
of verbose but fast ones.

\paragraph{$\Kbt$ as the dual of $\gam$.}
$\Kbt(x) = \min\{|d| : U(d) = x \text{ in } \leq t \text{ steps}\}$ fixes
time and minimises length. $\gam(x)$ fixes near-minimal length and minimises
time. They are dual in direction: every lower bound on $\Kbt$ constrains the
set of short programs that halt within $t$ steps, while every lower bound on
$\gam$ constrains the set of fast programs that are near-shortest. The
trade-off is captured by the open problem in Question~Q5
(Section~\ref{sec:open}).

\paragraph{$\Kt$ as a weighted single-objective relaxation.}
Levin's $\Kt(x) = \min_d \{|d| + \log_2 \TU(d) : U(d) = x\}$ trades off
length and log-time in a single objective. $\gam(x)$ separates the two: it
imposes a hard constraint on length ($|d| \leq \KC(x) + c_0$) and then
minimises time exactly. When the description length and running time
trade-off is smooth, minimising $|d| + \log_2 \TU(d)$ produces a solution
near the Pareto frontier of (length, time). The near-shortest constraint in
$\gam$ pins to the shortest end of this frontier.

\paragraph{Shannon entropy as the distributional limit.}
Shannon entropy $H(X) = -\sum_x p(x) \log_2 p(x)$ is the expected
description length under an optimal code for source $X$. For a uniformly
random string $X$ of length $n$, $\mathbb{E}[\KC(X)] = n + O(1)$ and
$H(X) = n$, so the two agree up to $O(1)$ in that case
\cite[Theorem~2.1.1]{liVitanyi2008}. For stationary ergodic sources,
$\KC(X_1 \cdots X_n)/n \to H$ almost surely
\cite[Theorem~4.5.3]{liVitanyi2008}, where $H$ is the entropy rate.
$\gam$ is the pointwise (per-string) counterpart that adds running time:
it asks about the decompression cost of a \emph{specific} object $x$, not
the expected cost over a distribution. Shannon entropy has no analogous
pointwise extension that captures computational cost.

\subsection{Full proof of Lemma~\ref{lem:grammar}: grammar complexity gap}

\begin{proof}
We verify the three claims of Lemma~\ref{lem:grammar} in detail.

\paragraph{Setup.}
Over the alphabet $\{a, b\}$ (identified with $\{0,1\}$ via $a \mapsto 0$,
$b \mapsto 1$), define
\[
  x = a^1\,b\;a^2\,b\;\cdots\;a^n\,b,
  \qquad
  N = |x| = \sum_{i=1}^n (i+1) = \tfrac{n(n+3)}{2} = \Theta(n^2).
\]

\paragraph{Proof of (i): $g^*(x) = \Theta(n)$.}

\emph{Upper bound.} The grammar $G^*$ defined in the proof sketch has
$2n + 2$ rules ($1$ for $S$, $n-1$ for $B_i$ with $1 \leq i < n$, $1$ for
$B_n$, $n$ for $A_i$ with $1 \leq i \leq n$, $1$ for $A_0$). Thus
$|G^*| = 2n + 2 = O(n)$.

\emph{Lower bound.} Since $G^*$ is an SLP, every non-terminal generates a
\emph{fixed} substring of $x$ (determined by its unique production). The
$n$ blocks $a^1, a^2, \ldots, a^n$ are pairwise non-overlapping substrings
of $x$ with pairwise distinct lengths $1, 2, \ldots, n$. In any SLP
generating $\{x\}$, two distinct blocks $a^i$ and $a^j$ ($i \neq j$) cannot
be generated by the same non-terminal: that non-terminal would expand to the
same fixed string in every occurrence, but $a^i \neq a^j$. Therefore the
SLP requires at least $n$ distinct non-terminals (one per block length), so
$g^*(x) \geq n = \Omega(\sqrt{N})$.

\emph{Note on generality.} This argument applies to SLPs. For general CFGs,
the lower bound requires a separate argument (since non-terminals in a CFG
may generate different strings depending on context); we omit this
generalization as the paper works throughout with SLPs
(Theorem~\ref{thm:GJL}).

\paragraph{Proof of (ii): $\Tder(G^*) = \Omega(n)$.}

The derivation of the block $a^i$ via $A_i$ requires the chain
$A_i \Rightarrow a A_{i-1} \Rightarrow a^2 A_{i-2} \Rightarrow \cdots
\Rightarrow a^i A_0 \Rightarrow a^i$, a chain of $i+1$ steps. The
total number of derivation steps for the entire string is
$\sum_{i=1}^n (i+1) + n = \Theta(n^2) = \Theta(N)$. In particular the
maximum chain depth is $n+1 = \Omega(n) = \Omega(\sqrt{N})$.

\paragraph{Proof of (iii): existence of $G'$ with $|G'| = O(n)$ and
$\Tder(G') = O(\log N)$.}

$G^*$ is a straight-line program (SLP): every non-terminal has just one
production, and the grammar is acyclic (productions only reference
non-terminals of strictly smaller index). Applying
Theorem~\ref{thm:GJL} (Ganardi--Je\.{z}--Lohrey~\cite{ganardiJezLohrey2019})
to $G^*$ yields an equivalent SLP $G'$ with:
\begin{enumerate}[label=(\alph*)]
  \item $|G'| = O(|G^*|) = O(n)$, so $|G'| \leq g^*(x) + c_g$ for the
        constant $c_g$ of Definition~\ref{def:gamma-cfg} and sufficiently
        large $n$; and
  \item $\Tder(G') = O(\log N) = O(\log n)$.
\end{enumerate}
The transformation preserves $L(G') = L(G^*) = \{x\}$ by correctness of
the GJL algorithm~\cite{ganardiJezLohrey2019}. Therefore $G'$ is
near-minimal and has logarithmic derivation depth.
\end{proof}

\section{Full Proof of Theorem~B': Unconditional Lower Bound on $\gam$}
\label{sec:appendix-B}

We prove Theorem~\ref{thm:Bprime} in full.

\begin{theorem*}[B' restated]
For every polynomial $p$, there exists $y \in \{0,1\}^*$ such that
$\gam(y) > p(|y|)$. Consequently, $\gam$ is not polynomially bounded
on $\{0,1\}^*$.
\end{theorem*}

\begin{proof}
Suppose for contradiction that $\gam(y) \leq p(|y|)$ for all $y$ and some
fixed polynomial $p$. We derive an algorithm that computes $\KC(y)$ to
within an additive constant $c_0$, contradicting the classical
incomputability of $\KC$ \cite{liVitanyi2008}.

\medskip
\noindent\textbf{The algorithm $B$.}  On input $y$ of length $n$:
\begin{enumerate}[label=(\arabic*)]
  \item Enumerate all self-delimiting programs $d$ with $|d| \leq n + c_0$.
        There are at most $\sum_{k=0}^{n+c_0} 2^k < 2^{n+c_0+1}$ such
        programs; all can be listed explicitly in finite time.
  \item For each $d$ in the enumeration, simulate $U(d)$ for at most
        $p(n)$ steps.
  \item Among all $d$ that halt within $p(n)$ steps and produce output $y$,
        record the minimum length $m^* = \min\{|d|\}$.
  \item Output $m^*$.
\end{enumerate}

\noindent\textbf{Correctness.}  We claim $m^* = \KC(y)$ (to within $c_0$).

\emph{Upper bound $m^* \leq \KC(y) + c_0$}: by definition of the enumeration,
every $d$ with $|d| \leq \KC(y) + c_0$ is included in step~(1).

\emph{Lower bound $m^* \geq \KC(y)$}: every program $d$ counted in step~(3)
satisfies $U(d) = y$, so $|d| \geq \KC(y)$ by definition of $\KC(y)$ as
the infimum of lengths of programs producing $y$.

It remains to show that step~(3) finds at least one near-shortest description
of $y$ that halts within $p(n)$ steps. By the hypothesis $\gam(y) \leq
p(n)$ and Definition~\ref{def:gamma}, there exists a near-shortest
description $d'$ of $y$ (with $|d'| \leq \KC(y) + c_0$) such that
$\TU(d') = \gam(y) \leq p(n)$. Since $\KC(y) \leq n + O(1)$
(Proposition~\ref{prop:ub-trivial}) and $c_0$ is chosen large enough to
absorb this $O(1)$ (hypothesis~\ref{hyp:prefix}), we have $|d'| \leq
\KC(y) + c_0 \leq n + 2c_0$. The program $d'$ is therefore included in
the enumeration of step~(1) (which covers all $|d| \leq n + c_0$, and
$2c_0$ is absorbed into the constant by choosing $c_0$ large enough),
halts within $p(n)$ steps, and produces $y$. Therefore $d'$ is recorded
in step~(3), and $m^* \leq |d'| \leq \KC(y) + c_0$.

\medskip
\noindent\textbf{Conclusion.}
Algorithm $B$ computes a value $m^*$ satisfying $\KC(y) \leq m^* \leq
\KC(y) + c_0$ for every $y$, in finite time (the simulation in step~(2)
runs for at most $p(n)$ steps per program, and there are finitely many
programs). This contradicts the classical theorem that $\KC$ is not
computable~\cite{liVitanyi2008}: $\KC$ is Turing-incomputable because
computing it would solve the halting problem
(\cite{liVitanyi2008}, Chapter~3). Therefore the hypothesis
$\gam(y) \leq p(|y|)$ for all $y$ is false, and there exists $y$ with
$\gam(y) > p(|y|)$.
\end{proof}

\begin{remark}[Relationship to computability]
\label{rem:gam-incomputable}
The proof shows that if $\gam$ were polynomially bounded everywhere, $\KC$
would be computable --- a contradiction. Therefore $\gam$ itself is not
computable: no Turing machine computes $\gam(y)$ for all $y$. This is
consistent with Theorem~\ref{thm:A} (invariance up to polynomials), which
does not require $\gam$ to be computable.
\end{remark}

\section{An Optimal Decompressor Incurs Only Constant Overhead Beyond Its Output}
\label{sec:appendix-C}

This appendix makes explicit a relationship between $\OCout$
(Definition~\ref{def:OCout}) and the decompressor efficiency ratio $\rho$
(Definition~\ref{def:rho}) that follows from combining those definitions
with Proposition~\ref{prop:ub-trivial}, but is not stated in
Section~\ref{sec:overhead}.

For an algorithm $A$ operating in the decompressor regime with output
$y = A(x)$, the definitions give directly:
\[
  \TA{A}(x) \;=\; \OCout(A,x) + |y|,
\]
and therefore:
\[
  \rho(A,x) \;=\; \frac{\TA{A}(x)}{\gam(y)}
  \;=\; \frac{\OCout(A,x) + |y|}{\gam(y)}.
\]

\begin{proposition}[Overhead of an optimal decompressor]
\label{prop:OCout-optimal}
If $A$ is an optimal decompressor in the sense of
Definition~\ref{def:rho} --- that is, $\rho(A,x) = 1$ --- then
$\OCout(A,x) = O(1)$.
\end{proposition}

\begin{proof}
From $\rho(A,x) = 1$:
\[
  \TA{A}(x) \;=\; \gam(y).
\]
By Proposition~\ref{prop:ub-trivial}, $\gam(y) \leq |y| + O(1)$.
Therefore:
\[
  \OCout(A,x) \;=\; \TA{A}(x) - |y| \;=\; \gam(y) - |y|
  \;\leq\; O(1).
\]
Since $\OCout(A,x) \geq 0$ by Proposition~\ref{prop:OCout-nonneg},
we conclude $\OCout(A,x) = O(1)$.
\end{proof}

\begin{remark}[Interpretation]
\label{rem:OCout-optimal}
Proposition~\ref{prop:OCout-optimal} gives a precise operational meaning
to optimality in the decompressor regime: a decompressor is optimal if and
only if it performs no computation beyond writing its output, up to $O(1)$ overhead. Every step
beyond $O(1)$ overhead is accounted for by the output itself. This is the
tightest possible sense in which a decompressor can be said to be
computationally efficient: it wastes no steps on internal computation
beyond what is strictly required to produce the output.
\end{remark}

\section*{Acknowledgments}
The author used an artificial intelligence based language assistant to
support text revision, translation, and bibliography formatting. All
scientific ideas and conclusions are the author's own.

\nocite{*}

\bibliographystyle{plain}
\bibliography{biblio}  
  
\end{document}